\documentclass[11pt]{amsart}

\usepackage{amsmath}

\usepackage{graphicx}
\usepackage{subfigure}
\newcommand{\subf}[2]{%
  {\small\begin{tabular}[t]{@{}c@{}}
  #1\\#2
  \end{tabular}}%
}

\usepackage{amssymb}

\usepackage{amsthm}

\usepackage{color}

\newtheorem{thm}{Theorem}[section]

\newtheorem{prop}{Proposition}[section]

\newtheorem{lem}{Lemma}[section]

\newtheorem{cor}{Corollary}[section]

\newtheorem{exa}{Example}[section]

\def\<{\left<}\def\>{\right>}
\def\({\left(}\def\){\right)}

\allowdisplaybreaks

\begin{document}

\title[Parisian excursion  for draw-down reflected L\'evy insurance risk process]{Parisian excursion with capital injection for draw-down reflected L\'evy insurance risk process}
\author{ Budhi Surya,\,\, Wenyuan Wang,  \,\,Xianghua Zhao\,\,and\,\,  Xiaowen Zhou }

\address{Budhi Surya:
School of Mathematics and Statistics, Victoria University of Wellington, New Zealand; Email address: budhi.surya@msor.vuw.ac.nz.
Wenyuan Wang: School of Mathematical Sciences, Xiamen University, Fujian 361005, People's Republic of China;
Email address: wwywang@xmu.edu.cn.
Xianghua Zhao: School of Mathematical Sciences, Qufu Normal University, Qufu 273165, China; Email address: qfzxh@163.com.
Xiaowen Zhou: Department of Mathematics and Statistics, Concordia University, Montreal, Quebec, Canada H3G 1M8; Email address:  xiaowen.zhou@concordia.ca.}



\begin{abstract}
This paper discusses Parisian ruin problem with capital injection for  L\'evy insurance risk process. Capital injection takes place at the draw-down time of the surplus process when it drops below a pre-specified function of its last record maximum. The capital is continuously paid to keep the surplus above the draw-down level until either the surplus process goes above the  record high or a Parisian type  ruin occurs, which is announced at the first instance the surplus process has undergone an excursion below the record for an independent exponential period of time consecutively since the time the capital was first injected. Some distributional identities concerning the excursion are presented. Firstly, we give the Parisian ruin probability and the joint Laplace transform (possibly killed at the first passage time above a fixed level of the surplus process) of the ruin time, surplus position at ruin, and the total capital injection at ruin. Secondly, we obtain the $q$-potential measure of the surplus process killed at Parisian ruin. Finally, we give expected present value of the total discounted capital payments up to the Parisian ruin time. The results are derived using recent developments in fluctuation and excursion theory of spectrally negative L\'evy process and are presented  semi explicitly in terms of the scale function of the L\'evy process. Some numerical examples are given to facilitate the analysis of the impact of initial surplus and frequency of observation period to the ruin probability and to the expected total capital injection.
\end{abstract}

\maketitle

\section{Introduction }
Parisian ruin has been actively studied since its introduction by Chesney et al. (1997). In contrary to the classical ruin model of Cram\'er-Lundberg in which case ruin/default occurs at the first instance the underlying (surplus) process crossing below a threshold, Parisian ruin is announced at the first time the process has undertaken an excursion below the default level for a fixed consecutive period of time. It has been applied in finance, among others, by Francois and Morellec (2004), Broadie et al. (2007), and recently by Antill and Grenadier (2019), in particular for studying firm's optimal capital structure in the presence of Chapters 7 and 11 (default and reorganization proceeding). The default level may be determined endogenously in the sense of Leland and Toft (1996) by maximizing the firm's equity value. Under Chapter 11 the firm is granted a dilution period during which the firm is given the opportunity to operate and reorganize itself until its asset value goes above the default level, or otherwise liquidated. In their recent work, Palmowski et al. (2020) revisit the Leland-Toft model under spectrally negative L\'evy process, discussed earlier in Hilberink and Rogers (2002) and Kyprianou and Surya (2007), by considering information of the firm's asset only available periodically at Poisson time. They showed using the results of Albrecher et al. (2016) that under Poisson observation of the firm's assets, nonzero credit spreads holds for firm with lower initial endowment than default level, and in particular, the default time corresponds to the Parisian ruin time with an independent exponential excursion period.

It was introduced to insurance/actuarial science literature started by the work of Dassios  and Wu (2008) for compound Poisson process, then extended to spectrally negative L\'evy process subsequently by Czarna and Palmowski (2011), Loeffen et al. (2013), among others. Further distributional identities concerning Parisian ruin with independent exponential excursion period were discussed in Baurdoux et al. (2016). The results are generalized to fixed excursion period for a class of penalty functions by Loeffen et al. (2018) which identify known results on Parisian ruin.

It is worth mentioning that the Parisian excursion discussed in the above literature gets started at the first passage below a threshold (default level) of the underlying process. In the past decades some discussions have been developed towards risk protection mechanism against certain financial assets' outperformance over their last record maximum, also referred to as high-water mark or draw-down. We refer interested readers to the works by Goetzmann et al. (2003) and Agarwal et al. (2009). Default is triggered when the underlying process has gone below a specified level from its last record maximum. Distributional identities regarding first-passage above a threshold for draw-down process were presented in Avram et al. (2004, 2007) and used for pricing Russian options under randomized maturity, and solving optimal dividend problem where the cumulative paid dividend is given by the running supremum of the surplus process. First-passage identities for draw-down process were later extended to a more general form of threshold boundary for e.g. by Zhou (2007), Wang and Zhou (2018), and Li et al. (2019). Parisian excursion below a fixed level from the last record maximum of the surplus process was considered in Surya (2019).

The introduction of capital injection to the firm may prevent the firm from going default at the time its asset value decreases below a certain threshold and the firm needs to meet its commitment to pay dividends to the shareholders. We refer among others to Kulenko and Schmidli (2008), Tao et al. (2011), Avanzi et al. (2011), Bayraktar et al. (2013), and Wang et al. (2019).

In this paper we study Parisian ruin from the last record maximum of surplus process with capital injection. We assume there is no dividend payment in the model. As the source of uncertainty in the surplus process is the downward jumps of the L\'evy insurance risk process. Capital is provided to the firm by the stakeholder as soon as the surplus process goes below a given  function of the last record maximum of the surplus process. It is continuously provided to the firm to keep the surplus above the draw-down level  until the surplus goes back to above the last record or the ruin occurs. If the surplus process has stayed below the record since the first capital injection longer than an independent exponential period of time, the firm faces the credit event and is liquidated. Distributional identities concerning the Parisian ruin are presented.

\setcounter{section}{1}

The rest of the paper is arranged as follows. In Section \ref{2} we introduce the model and formulate the problems we are interested. The main results and proofs are presented in Section \ref{3}. Section \ref{4} presents some numerical examples of the main results. In particular, they are presented to analyze the impact of observation frequency and initial value of surplus to various shapes of ruin probability and expected nett present value of the total capital injection. Section \ref{5} concludes this paper.

\section{Spectrally negative L\'{e}vy process and its reflected processes}
\setcounter{section}{2}\label{2}

Write $X\equiv\{X(t);t\geq0\}$, defined on a probability space $(\Omega, \mathcal{F},(\mathcal{F}_t)_{t\geq 0},\mathbb{P})$ with probability laws $\{\mathbb{P}_{x};x\in\mathbb{R}\}$ denoting the family of probability laws of $X$ such that $X_0=x$, with $\mathbb{P}=\mathbb{P}_0$, and natural filtration $\{\mathcal{F}_{t};t\geq0\}$  of $X$ satisfying the usual conditions of right-continuity and completeness. In particular, we exclude a spectrally negative L\'{e}vy process that is not a purely increasing linear drift or the negative of a subordinator. The L\'evy-It\^o sample paths decomposition of the L\'evy process is given by
\begin{equation}\label{eq:LevyIto}
\begin{split}
X_t=\mu t + \sigma B_t &+\int_0^t\int_{\{x<-1\}} x\mathcal{N}(\mathrm{d}x,\mathrm{d}s) + \int_0^t\int_{\{-1\leq x <0\}} x\big(\mathcal{N}(\mathrm{d}x,\mathrm{d}s)-\nu(\mathrm{d}x)\mathrm{d}s\big),
\end{split}
\end{equation}
where $\mu\in\mathbb{R}$, $\sigma\geq0$ and $(B_t)_{t\geq0}$ is standard Brownian motion, whilst $\mathcal{N}(\mathrm{d}x,\mathrm{d}t)$ denotes the Poisson random measure associated with the jumps process $\Delta X_t:=X_t-X_{t-}$ of $X$. This Poisson random measure has compensator given by $\nu(\mathrm{d}x)\mathrm{d}t$, where $\nu$ is the L\'evy measure satisfying the condition:
\begin{equation}\label{eq:intcond}
\int_{-\infty}^0 (1\wedge x^2)\nu(\mathrm{d}x)<\infty.
\end{equation}
In the expression (\ref{eq:LevyIto}), $X$ may define the surplus process of an insurance firm in which $\mu$ represents the premium rate charged on the insurance holder and $\sigma B_t$ is of volatile trading uncertainties which results from the firm investing in financial market. The other two jump components correspond to the compensated small claims and uncompensated big claims from the insurance holder whose distribution is given by the L\'evy measure $\nu$ satisfying the integral constraint (\ref{eq:intcond}). In the classical model of Cramer-Lundberg risk process with positive drift $c>0$, $c= \mu -\int_{-\infty}^0 x \mathbf{1}_{\{x>-1\}} \nu(\mathrm{d}x)$ with $\nu(\mathrm{d}x)=\beta F(\mathrm{d}x)$ for Poisson claim arrival intensity $\beta>0$ and distribution $F$ of the claim size.

Denote the running supremum process $\overline{X}\equiv\{\sup\limits_{0\leq s\leq t}X(s),\,t\geq0\}$ with $\overline{X}(
0)=x$ under $\mathbb{P}_{x}$.
Given a value  $a\in\mathbb{R}$, the process $X$ reflected from below at the level $a$ is defined as
\[X(t)-(\underline{X}(t)-a)\wedge 0, \,\, t\geq 0 \]
where $\underline{X}(t):=\inf_{0\leq s\leq t}X(s)$ with $\underline{X}(
0)=x$ under $\mathbb{P}_{x}$, denotes the running infimum process.
Let $\{Y(t), t\geq0\}$ be the process $X$ reflected from below at the level $0$ (cf., Pistorius (2004)).

The  draw-down time associated to a draw-down function $\xi$ on $(-\infty,\infty)$ satisfying $\xi(x)<x, \,\, x\in(-\infty,\infty)$,  the $\xi$-draw-down time in short, is defined as
$$\tau_{\xi}\equiv \tau_{\xi}(X):=\inf\{t\geq0: X(t)<\xi(\overline{X}(t))\}$$
with the convention $\inf\emptyset:=\infty$.
 We define the process $X$ reflected at the $\xi$-draw-down time $\tau_{\xi}$ as
  \[X(t)-\mathbf{1}_{[\tau_\xi,\infty)}(t)\left(\inf_{\tau_\xi\leq s\leq t}{X}(s)- \xi(\overline{X}(\tau_\xi))\right)\wedge 0, \,\, t\geq 0, \]
where we call $\xi(\overline{X}(\tau_\xi))$ the draw-down level at the draw-down time $\tau_\xi$.


We now define the draw-down reflected process $U$ for $X$.
Intuitively, the process $U$ initially agrees with $X$ until the first draw-down time of $U$. Then it starts to evolve according to $X$ reflected at the draw-down level until the next draw-down time of $U$ when it is reflected at the draw-down level again, and so on. Then given that $U(s)=\overline{U}(s):=\sup_{0\leq t\leq s}U(t)$, the process $\{U(t);t\geq s\}$ evolves without reflection until the next draw-down time $\tau_\xi$; and given that $\overline{U}(s)> U(s)$,
the process $\{U(t); t\geq s\}$ is reflected from below at the current draw-down level $\xi(\overline{U}(s))$ until it comes back to the level  $\overline{U}(s)$. Note that the process $U$ is not a Markov process in general, but the process $(U,\overline{U})$ is Markovian.
Write $\mathbb{P}_{x,y}$ and $\mathbb{E}_{x,y}$ for the law of $(U,\overline{U})$ such that $U(0)=x$ and $\overline{U}(0)=y$. For simplicity, denote $\mathbb{P}_{x}=\mathbb{P}_{x,x}$ and $\mathbb{E}_{x}=\mathbb{E}_{x,x}$.

To be more precise, define $T_0:=0$ and $U(T_0):=X(0) $.
 Suppose first that  for $n\geq 1$, $U(t)$ has been defined on $[0, T_n]$ for $T_n<\infty$, $n\geq 1$.
 Let $X_{n+1}$ be an independent copy of $X$ starting at $U(T_n)$ and $U_{n+1}$ be the process $X_{n+1}$ reflected at its $\xi$-draw-down time $\tau_{\xi}({X}_{n+1})$. If $\tau_{\xi}(X_{n+1})=\infty$, let $T_{n+1}:=\infty$, and if $\tau_{\xi}(X_{n+1})<\infty$, let
$$T_{n+1}:=T_n+\inf\{t\geq 0: U_{n+1}(t)>\overline{X}_{n+1}(\tau_{\xi}(X_{n+1})) \}, $$
 where $\overline{X}_{n+1}(t):=\sup\limits_{0\leq s\leq t}X_{n+1}(s)$. Observe that $T_{n+1}<\infty$ if $\tau_{\xi}(X_{n+1})<\infty $.
Then define
\[U(T_n+t):=U_{n+1}(t) \quad\text{for}\quad t\in [0, T_{n+1}-T_n) \quad\text{and}\quad U(T_{n+1}):=U_{n+1}(T_{n+1}-T_n)\quad\text{if}\quad T_{n+1}<\infty.\]
Suppose now that $U(t)$ has been defined on $[0, T_n=\infty)$ for $ n\geq 0$.
For convenience, let $T_{n+1}:=\infty$.
For the well-definedness  of the process $U$, we are referred to Wang and Zhou (2019).

For the process $X$, define its first up-crossing time of level $b\in(-\infty,\infty)$ and first down-crossing time of level $c\in(-\infty,\infty)$, respectively, by
\begin{eqnarray}
\tau^{+}_{b}:=\inf\{t\geq0: X(t)>b\}\,\,\, \text{and}\,\,\, \tau_{c}^{-}:=\inf\{t\geq0: X(t)<c\}.\nonumber
\end{eqnarray}

For the processes $Y$ and $U$, their first up-crossing times of $b\in(-\infty,\infty)$ are defined respectively by
\begin{eqnarray}
\sigma^{+}_{b}:=\inf\{t\geq0: Y(t)>b\} \,\,\, \text{and} \,\,\,\kappa_{b}^{+}:=\inf\{t\geq0: U(t)>b\}.\nonumber
\end{eqnarray}
The Parisian ruin time $\theta_{\lambda}$ of the process $U$ is defined as
\begin{eqnarray}
N_{\lambda}:=\inf\{n\geq 1: \,T_{n}-T_{n-1}-\tau_{\xi}(X_{n})>e_{\lambda}^{(n)}\},\quad \theta_{\lambda}:=T_{N_{\lambda}-1}+\tau_{\xi}(X_{N_{\lambda}})+e_{\lambda}^{(N_{\lambda})},\nonumber
\end{eqnarray}
where $\{e_{\lambda}^{(n)};n\geq 1\}$ is a sequence of i.i.d. exponential randoms with parameter $\lambda>0$, and is independent of $X$.

Due to the absence of positive jumps, it is therefore sensible to define
\begin{eqnarray}
\psi(\theta):=\ln \mathbb{E}_{x}\left(\mathrm{e}^{\theta (X_{1}-x)}\right)=\mu\theta+\frac{1}{2}\sigma^{2}\theta^{2}+\int_{(-\infty,0)}\left(\mathrm{e}^{\theta x}-1-\theta x\mathbf{1}_{(-1,0)}(x)\right)\nu(\mathrm{d}x),\nonumber
\end{eqnarray}
It is known that $\psi(\theta)$ is finite for  $\theta\in[0,\infty)$ in which case it is strictly convex and infinitely differentiable.
As in Bertoin (1996), the $q$-scale functions $\{W_{q};q\geq0\}$ of $X$ are defined as follows. For each $q\geq0$, $W_{q}:\,[0,\infty)\rightarrow[0,\infty)$ is the unique strictly increasing and continuous function with Laplace transform
\begin{eqnarray}
\int_{0}^{\infty}\mathrm{e}^{-\theta x}W_{q}(x)\mathrm{d}x=\frac{1}{\psi(\theta)-q},\quad \mbox{for }\theta>\Phi(q),\label{eq:scale}
\end{eqnarray}
where $\Phi(q)$ is the largest solution of the equation $\psi(\theta)=q$. Further define $W_{q}(x)=0 $ for $x<0$, and write $W$ for the $0$-scale function $W_{0}$.

It is known that $W_{q}(0)=0 $ if and only if process $X$ has sample paths of unbounded variation.
If $X$ has sample paths of unbounded variation, or if $X$ has sample paths of bounded variation and the L\'{e}vy measure has no atoms, then the scale function $W_{q}$ is continuously differentiable over $(0, \infty)$.
By Loeffen (2008), if $X$ has a L\'{e}vy measure which has a completely monotone density, then $W_{q}$ is twice continuously differentiable over $(0, \infty)$ when $X$ is of unbounded variation.
Moreover, if process $X$ has a nontrivial Gaussian component, then $W_{q}$ is twice continuously differentiable over $(0, \infty)$. Below is an example of the scale function $W_q$ for downward jump-diffusion process with exponentially distributed jumps, which will be used for numerical examples discussed in Section 4.

\begin{figure}[t!]
\centering
\begin{tabular}{cc}
\subf{\includegraphics[width=0.4\textwidth]{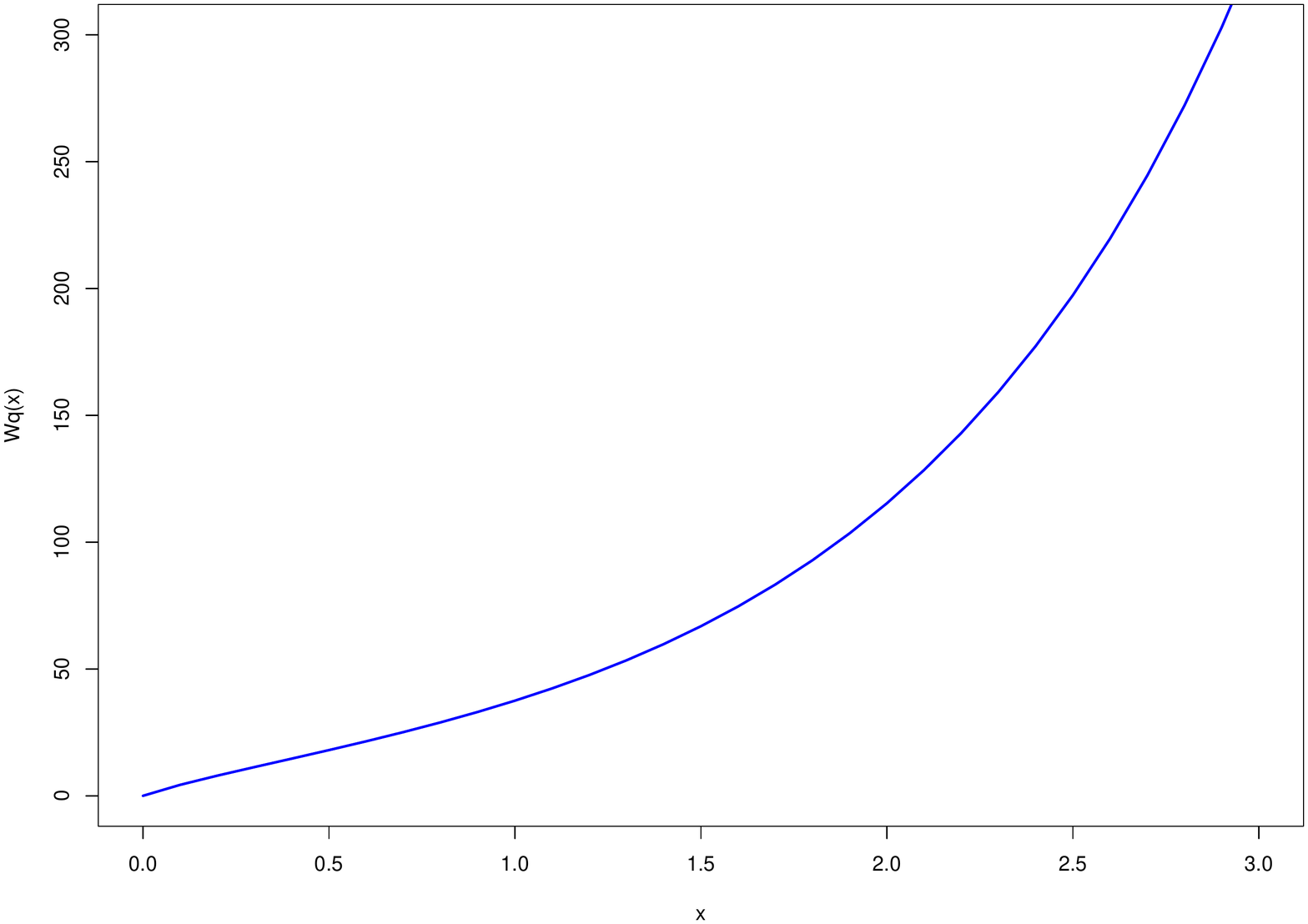}}
       {(a)  $W_q(x)$ for $\sigma=0.2$.}

&
\subf{\includegraphics[width=0.4\textwidth]{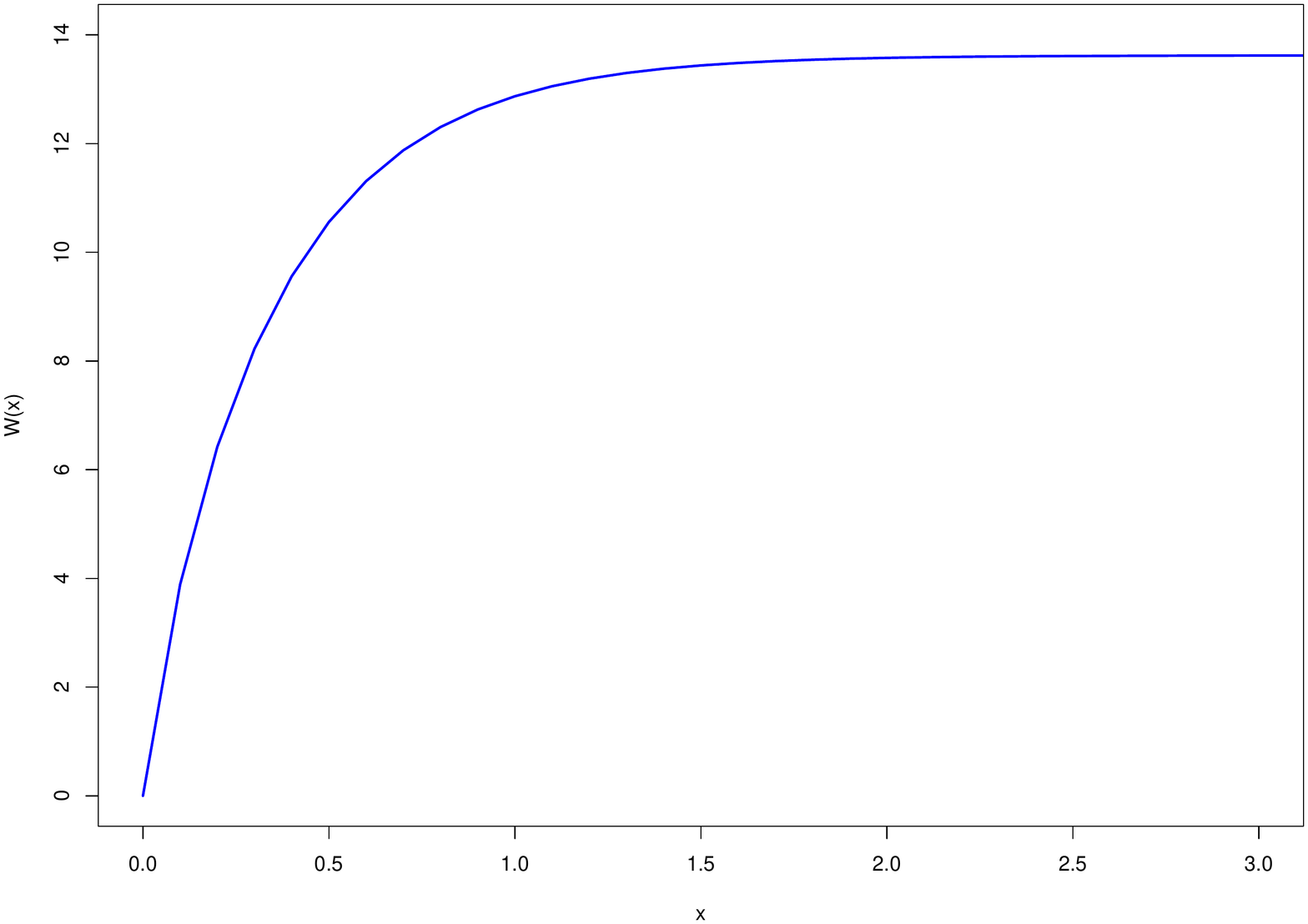}}
       {(b)  $W^{\Phi(q)}(x)$ for $\sigma=0.2$.}
\\
\subf{\includegraphics[width=0.4\textwidth]{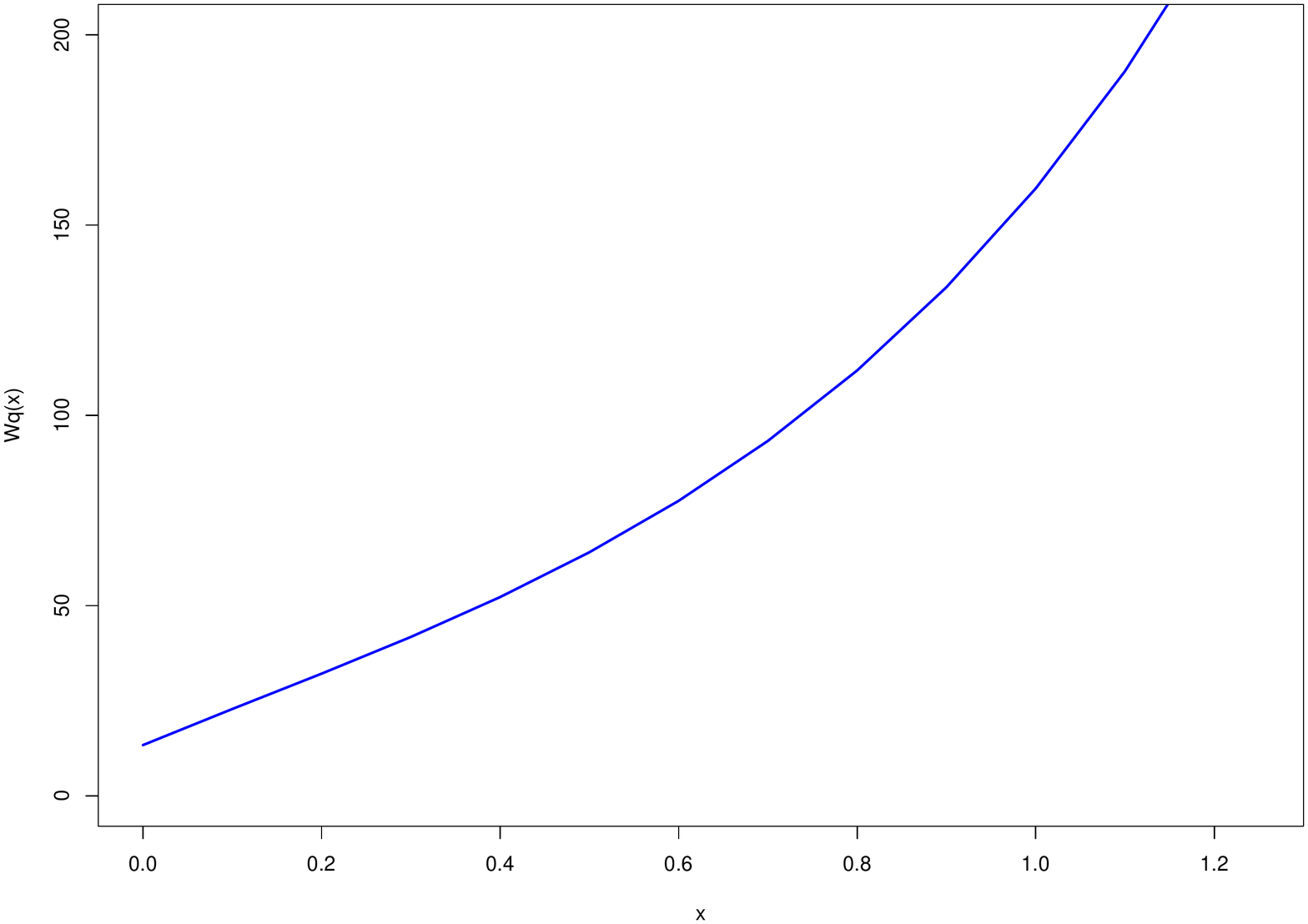}}
       {(c)  $W_q(x)$ for $\sigma=0$.}
&
\subf{\includegraphics[width=0.4\textwidth]{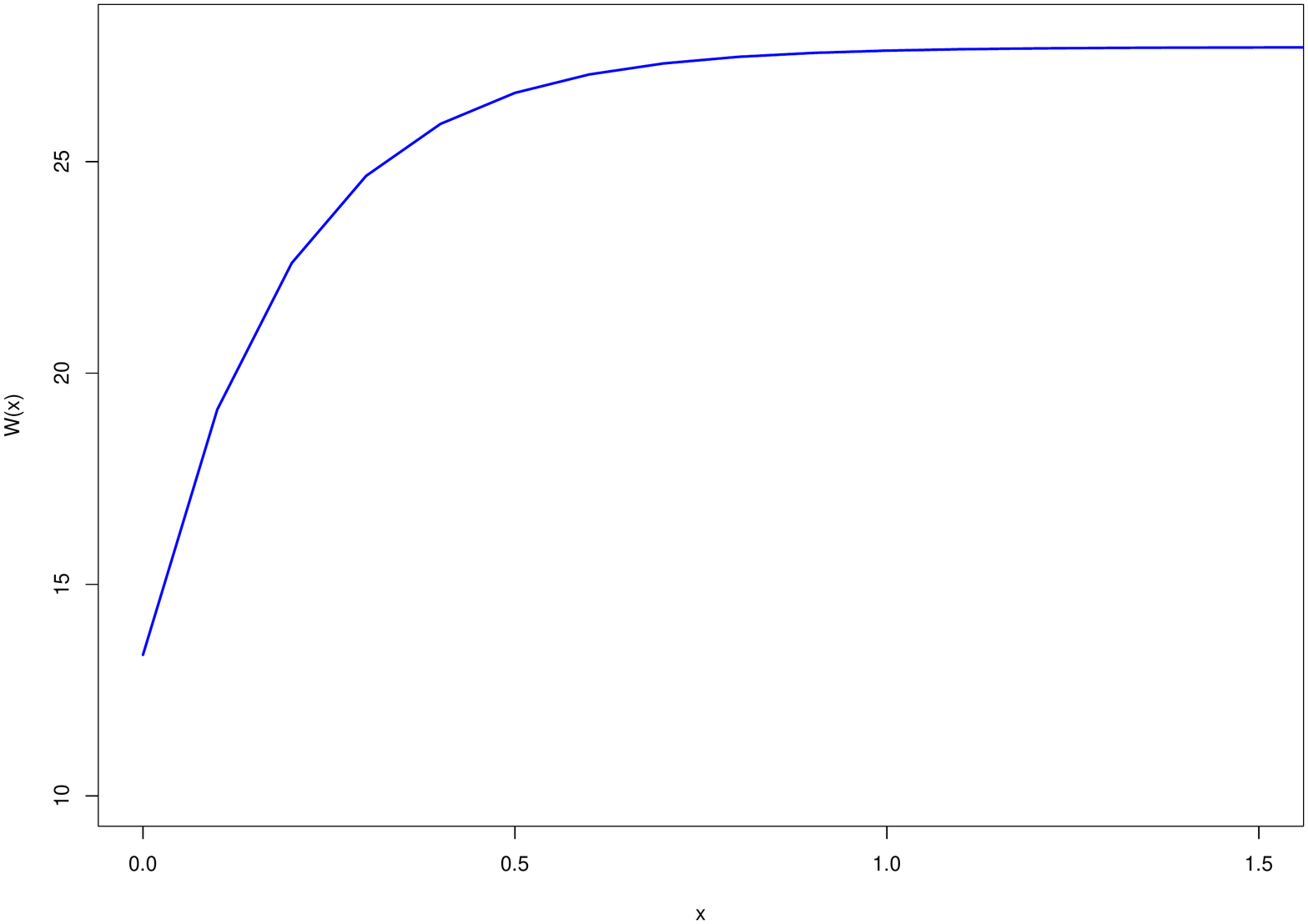}}
       {(d) $W^{\Phi(q)}(x)$ for $\sigma=0$.}
\\
\end{tabular}
\caption{Scale function $W_q(x)$ and $W^{\Phi(q)}(x)$ for downward jump-diffusion process with Laplace exponent $\psi(\theta)=\mu\theta+\frac{\sigma^2}{2}\theta^2-\frac{a\theta}{\theta+c}$ for $\mu=0.075, a=0.5, c=9, q=0.05$. }\label{fig:SF}
\end{figure}

\begin{exa}\label{ex:Ex1}\rm
Consider one-sided jump-diffusion process $X$ with $\psi(\theta)=\mu\theta+\frac{\sigma^2}{2}\theta^2-\frac{a\theta}{\theta+c}$ for all $\theta\in\mathbb{R}$ s.t. $\theta\neq -c$. It is known that the inverse of the Laplace transform (\ref{eq:scale}) for $q>0$ is
\begin{align}\label{eq:ScaleF}
W_q(x)=\frac{\mathrm{e}^{-\beta_2 x}}{\psi^{\prime}(-\beta_2)} + \frac{\mathrm{e}^{-\beta_1 x}}{\psi^{\prime}(-\beta_1)} + \frac{\mathrm{e}^{\Phi(q)x}}{\psi^{\prime}(\Phi(q))}, \quad \forall x\geq 0,
\end{align}
where $-\beta_1$, $-\beta_2$, and $\Phi(q)$ denotes three roots of $\psi(\theta)=q$ s.t. $-\beta_2<-c<-\beta_1<0<\Phi(q)$. It is straightforward to check by taking Laplace transform that $W_q(x)=\mathrm{e}^{\Phi(q)x}W^{\Phi(q)}(x)$ where
\begin{align}\label{eq:ScaleF2}
W^{\Phi(q)}(x)=\frac{\mathrm{e}^{-(\beta_2+\Phi(q)) x}}{\psi^{\prime}(-\beta_2)} + \frac{\mathrm{e}^{-(\beta_1+\Phi(q)) x}}{\psi^{\prime}(-\beta_1)} + \frac{1}{\psi^{\prime}(\Phi(q))}, \quad \forall x\geq 0,
\end{align}
with $W^{\Phi(q)}(x)=0$ for $x<0$. It is known that $\psi^{\prime}(-\beta_2)<0$, $\psi^{\prime}(-\beta_1)<0$ and $\psi^{\prime}(\Phi(u))>0$. In fact $W^{\Phi(q)}(x)$ plays the role of $W(x)$ under the Esscher transform of measure $\mathbb{P}^{\Phi(q)}$ defined by $\frac{\mathrm{d}\mathbb{P}^{\Phi(q)}}{\mathrm{d}\mathbb{P}}\Big\vert_{\mathcal{F}_t}=\mathrm{e}^{\Phi(q) X_t -qt}$. It is straight forward to check that $W^{\Phi(q)}(x)$ is increasing for $x\geq 0$, and so is $W_q(x)$, concave and is bounded from above by $1/\psi^{\prime}(\Phi(q))$.
\end{exa}

The interested readers are referred to Chan et al. (2011) and Kuznetsov et al. (2012) for more detailed discussions on the smoothness of scale functions.
For results on numerical computation of the scale function, the readers are referred to Surya (2008), Hubalek and Kyprianou (2011) and the references therein.

 Further define
\begin{eqnarray}
Z_{q}(x):=1+q\int_{0}^{x}W_{q}(z)\mathrm{d}z,\quad x\geq0,\nonumber
\end{eqnarray}
and
\begin{eqnarray}
Z_{q}(x,\theta):=\mathrm{e}^{\theta x}\left(1-\left(\psi(\theta)-q\right)\int_{0}^{x}\mathrm{e}^{-\theta z}W_{q}(z)\mathrm{d}z\right),\quad \theta\geq0, \,x\geq0,\nonumber
\end{eqnarray}
with $Z(x,\theta):=Z_{0}(x,\theta)$,
and
$$\overline{W}_{q}(x):=\int_{0}^{x}W_{q}(z)\mathrm{d}z,\quad q\geq0, x\geq0,$$
and
$$\overline{Z}_{q}(x):=\int_{0}^{x}Z_{q}(z)\mathrm{d}z=x+q\int_{0}^{x}\int_{0}^{z}W_{q}(w)\mathrm{d}w\mathrm{d}z,\quad q\geq0, x\geq0.$$

In the sequel, without loss of generality we assume $X_{1}\equiv X$. By Li et al. (2017), we have
\begin{eqnarray}\label{part1}
\mathbb{E}_{x}(\mathrm{e}^{-q\kappa_{b}^{+}}\mathbf{1}_{\{\kappa_{b}^{+}<\tau_{\xi}\}})=
\mathbb{E}_{x}\left(\mathrm{e}^{-q\tau_{b}^{+}}\mathbf{1}_{\{\tau_{b}^{+}<\tau_{\xi}\}}\right)
=\exp\left(-\int_{x}^{b}\frac{W_{q}^{\prime}(\overline{\xi}\left(z\right))}
{W_{q}(\overline{\xi}\left(z\right))}\mathrm{d}z\right),
\end{eqnarray}
where $\overline{\xi}(z)=z-\xi(z)$.
For $x\in[0,b]$ and $q\geq0$, from Proposition 2 in Pistorius (2004) we have
\begin{eqnarray}\label{two.sid.exit.Y}
\mathbb{E}_{x}(\mathrm{e}^{-q\sigma^{+}_{b}})=\frac{Z_{q}(x)}{Z_{q}(b)}.
\end{eqnarray}

By Kyprianou (2006), the resolvent measure corresponding to $X$ is absolutely continuous with respect to the Lebesgue measure with  density given by
\begin{eqnarray}\label{h2}
\hspace{-0.3cm}&&\hspace{-0.3cm}\int_{0}^{\infty}\mathrm{e}^{-qt}\mathbb{P}_{x}(X(t)\in \mathrm{d}y;t<\tau_{c}^{-}\wedge \tau_{b}^{+})\mathrm{d}t
\nonumber\\
\hspace{-0.3cm}&=&\hspace{-0.3cm}
\left(\frac{W_{q}(x-c)}{W_{q}(b-c)}W_{q}(b-y)-W_{q}(x-y)\right)\mathbf{1}_{(c,b)}(y)\mathrm{d}y,
\end{eqnarray}
for $x\in(c, b)$.
By Pistorius (2004), the resolvent measure corresponding to $Y$ is also absolutely continuous with respect to the Lebesgue measure and has a version of density given by
\begin{eqnarray}\label{reso.meas.Y}
\hspace{-0.3cm}&&\hspace{-0.3cm}\int_{0}^{\infty}\mathrm{e}^{-qt}\mathbb{P}_{x}(Y(t)\in \mathrm{d}y,t<\sigma^{+}_{b})\mathrm{d}t
\nonumber\\
\hspace{-0.3cm}&=&\hspace{-0.3cm}\left(\frac{Z_{q}(x)}{Z_{q}(b)}W_{q}(b-y)-W_{q}(x-y)\right)
\mathbf{1}_{[0,b)}(y)\mathrm{d}y,
\end{eqnarray}
where $x\in[0, b)$.

Define the total amount of capital injections made until time $t$ for the draw-down reflected process as
\begin{eqnarray}
R(t)\hspace{-0.3cm}&:=&\hspace{-0.3cm}
-\sum_{k=1}^{N-1}\mathbf{1}_{[T_{k-1}+\tau_{\xi}(X_k),\infty)}(t)\left(   \inf_{\tau_\xi(X_k)\leq s\leq T_k\wedge t-T_{k-1}}{X}_{k}(s)-\xi(\overline{X}_{k}(\tau_{\xi}(X_k)))\right)\wedge 0.
\nonumber
\end{eqnarray}
where $N:=\inf\{n: T_n=\infty\}=\inf\{n: \tau_\xi(X_{n})=\infty\}$.

In this paper, we are interested in evaluating:
\begin{itemize}
\item[(a)]
Expectation of the net present value of principal payment of one unit at time $\kappa_b^+\wedge\theta_{\lambda}$:
\begin{eqnarray}
U_{\xi}(x;b)=\mathbb{E}_x\Big(\mathrm{e}^{-q(\kappa_b^+\wedge\theta_{\lambda})}\Big). \nonumber
\end{eqnarray}
%
\item[(b)] The joint Laplace transform of $\kappa_{b}^{+}\wedge \theta_{\lambda}$, $U(\kappa_{b}^{+}\wedge \theta_{\lambda})$ and $R(\kappa_{b}^{+}\wedge \theta_{\lambda})$, i.e.
\begin{eqnarray}
\quad \quad G_\xi(x;b)
\hspace{-0.3cm}&
=&\hspace{-0.3cm}
\mathbb{E}_{x}\left(\mathrm{e}^{-q \left(\kappa_{b}^{+}\wedge \theta_{\lambda}\right)+u U\left(\kappa_{b}^{+}\wedge \theta_{\lambda}\right)-v R\left(\kappa_{b}^{+}\wedge \theta_{\lambda}\right)}\right),\quad q,\lambda, u,v\in \mathbb{R}_{+},\,  b\in\mathbb{R},\, x\in(-\infty,b]
.\nonumber
\end{eqnarray}
\item[(c)]
The potential measure of $U$ involving the Parisian ruin time, i.e.
\begin{eqnarray}
\label{}
\int_{0}^{\infty}\mathrm{e}^{-qt}\mathbb{P}_{x}\left(U(t)\in \mathrm{d}u,t<\kappa_{b}^{+}\wedge \theta_{\lambda}\right)\mathrm{d}t,\quad q,\lambda\in \mathbb{R}_{+},\,  b\in\mathbb{R},\,x,u\in(-\infty,b].
\nonumber
\end{eqnarray}
\item[(d)]
The expectation of the total discounted capital injections until $\kappa_{b}^{+}\wedge \theta_{\lambda}$, i.e.
\begin{eqnarray}
V_{\xi}(x;b)\hspace{-0.3cm}&=&\hspace{-0.3cm}\mathbb{E}_{x}\left(\int_{0}^{\kappa_{b}^{+}\wedge \theta_{\lambda}}\mathrm{e}^{-q t}\mathrm{d}R(t)\right),\quad q,\lambda\in \mathbb{R}_{+},\,  b\in\mathbb{R},\,x\in(-\infty,b].
\nonumber
\end{eqnarray}

\end{itemize}

We also briefly recall concepts in excursion theory for the reflected process $\{\overline{X}(t)-X(t);t\geq0\}$, and we refer to Bertoin (1996) for more details.
For $x\in(-\infty,\infty)$, the process $\{L(t):= \overline{X}(t)-x, t\geq0\}$ serves as a local time at $0$ for
the Markov process $\{\overline{X}(t)-X(t);t\geq0\}$ under $\mathbb{P}_{x}$.
Let the corresponding inverse local time be defined as
$$L^{-1}(t):=\inf\{s\geq0: L(s)>t\}=\sup\{s\geq0: L(s)\leq t\}.$$
Further let $L^{-1}(t-):=\lim\limits_{s\uparrow t}L^{-1}(s)$.
Define a Poisson point process $\{(t, e_{t}); t\geq0\}$ as
$$e_{t}(s):=X(L^{-1}(t))-X(L^{-1}(t-)+s), \,\,s\in(0,L^{-1}(t)-L^{-1}(t-)],$$
whenever the lifetime of $e_{t}$ is positive, i.e. $L^{-1}(t)-L^{-1}(t-)>0$.
Whenever $L^{-1}(t)-L^{-1}(t-)=0 $, define $e_{t}:=\Upsilon$ with $\Upsilon$ being an additional isolated point.
A result of It\^{o} states that $e$ is a Poisson point process with
characteristic measure $n$
if $\{\overline{X}(t)-X(t);t\geq0\}$ is recurrent; otherwise $\{e_{t}; t\leq L(\infty)\}$ is a Poisson point process stopped at the first excursion of infinite lifetime. Here, $n$ is a measure on the space  $\mathcal{E}$ of excursions,
i.e. the space $\mathcal{E}$  of c\`{a}dl\`{a}g functions $f$ satisfying
\begin{eqnarray}
&&f:\,(0,\zeta)\rightarrow (0,\infty)\,\quad \mbox{for some } \zeta=\zeta(f)\in(0,\infty],
\nonumber\\
&&f:\,\{\zeta\}\rightarrow (0,\infty)\,\,\,\,\,\,\quad \mbox{if } \zeta<\infty,
\nonumber
\end{eqnarray}
where $\zeta=\zeta(f)$ is the excursion length or lifetime; see Definition 6.13 of Kyprianou (2006) for the definition of $\mathcal{E}$.
Denote by $\varepsilon(\cdot)$, or $\varepsilon$ for short, a generic excursion
belonging to the space $\mathcal{E}$ of canonical excursions.
The excursion height of a canonical excursion $\varepsilon$ is
denoted by $\overline{\varepsilon}=\sup\limits_{t\in[0,\zeta]}\varepsilon(t)$. The first passage time of a canonical excursion $\varepsilon$ is defined
by
$$
\rho_{b}^{+}\equiv\rho_{b}^{+}(\varepsilon) :=\inf\{t\in[0,\zeta]: \varepsilon(t)>b\},
$$
with the convention $\inf\emptyset:=\zeta$.

Denote by $\varepsilon_{g}$
the excursion (away from $0$)  with left-end point $g$ for the reflected process $\{\overline{X}(t)-X(t);t\geq0\}$, and by $\zeta_{g}$ and $\overline{\varepsilon}_{g}$  the excursion's lifetime and the excursion's height, respectively; see Section IV.4 of Bertoin (1996).

\section{Main results}
\setcounter{section}{3}\label{3}

In this section we present the main results concerning the general draw-down reflected process $U$ with Parisian stopping.
For preparation, we recall the following result of Wang and Zhou (2019).

\begin{lem}\label{lemma2}
Given $\theta,\,q\in(0,\infty)$ and measurable function $\phi:\,(-\infty,\infty)\rightarrow(-\infty,\infty)$, we have
\begin{eqnarray}\label{10}
\hspace{-0.3cm}&&\hspace{-0.3cm}
\mathbb{E}_{x}\left(\mathrm{e}^{-q \tau_{\xi}}\,\mathrm{e}^{\theta X(\tau_{\xi})}\,\phi\left(\overline{X}(\tau_{\xi})\right); \tau_{\xi}<\tau_{b}^{+}\right)
\nonumber\\
\hspace{-0.3cm}&=&\hspace{-0.3cm}
\int_{x}^{b}\phi\left(s\right)\mathrm{e}^{\theta \xi(s)}
\exp\left(-\int_{x}^{s}\frac{W_{q}^{\prime}(\overline{\xi}\left(z\right))}
{W_{q}(\overline{\xi}\left(z\right))}\mathrm{d}z\right)
\nonumber\\
\hspace{-0.3cm}&&\hspace{-0.3cm}
\times\left(\frac{W_{q}^{\prime}(\overline{\xi}(s))}{W_{q}(\overline{\xi}(s))}Z_{q}(\overline{\xi}(s),\theta )-\theta  Z_{q}(\overline{\xi}(s),\theta )-(q-\psi(\theta ))W_{q}(\overline{\xi}(s))\right)
\mathrm{d}s,\quad x\in(-\infty, b].
\end{eqnarray}
In particular, we have
\begin{eqnarray}\label{12}
\hspace{-0.3cm}&&\hspace{-0.3cm}
\mathbb{E}_{x}\left(\mathrm{e}^{-q \tau_{\xi}}\,\phi\left(\overline{X}(\tau_{\xi})\right); \tau_{\xi}<\tau_{b}^{+}\right)
=
\int_{x}^{b}\phi\left(s\right)
\exp\left(-\int_{x}^{s}\frac{W_{q}^{\prime}(\overline{\xi}\left(z\right))}
{W_{q}(\overline{\xi}\left(z\right))}\mathrm{d}z\right)
\nonumber\\
\hspace{-0.3cm}&&\hspace{3.8cm}
\times\left(\frac{W_{q}^{\prime}(\overline{\xi}(s))}{W_{q}(\overline{\xi}(s))}Z_{q}(\overline{\xi}(s))-qW_{q}(\overline{\xi}(s))\right)
\mathrm{d}s,\quad x\in(-\infty, b]
,
\end{eqnarray}
and
\begin{eqnarray}\label{13}
\hspace{-0.3cm}&&\hspace{-0.3cm}
\mathbb{E}_{x}\left(\mathrm{e}^{\theta X(\tau_{\xi})}\,\phi\left(\overline{X}(\tau_{\xi})\right); \tau_{\xi}<\tau_{b}^{+}\right)
=
\int_{x}^{b}\phi\left(s\right)\mathrm{e}^{\theta \xi(s)}
\exp\left(-\int_{x}^{s}\frac{W^{\prime}(\overline{\xi}\left(z\right))}
{W^{}(\overline{\xi}\left(z\right))}\mathrm{d}z\right)
\nonumber\\
\hspace{-0.3cm}&&\hspace{3cm}
\times\left(\frac{W^{\prime}(\overline{\xi}(s))}{W^{}(\overline{\xi}(s))}Z^{}(\overline{\xi}(s),\theta )-\theta  Z^{}(\overline{\xi}(s),\theta )+\psi(\theta )W^{}(\overline{\xi}(s))\right)
\mathrm{d}s,\quad x\in(-\infty, b]
,
\end{eqnarray}
and
\begin{eqnarray}\label{14}
\hspace{-0.3cm}&&\hspace{-0.3cm}
\mathbb{E}_{x}\left(\mathrm{e}^{-q \tau_{\xi}}\left(\xi\left(\overline{X}(\tau_{\xi})\right)-X(\tau_{\xi})\right); \tau_{\xi}<\tau_{b}^{+}\right)
=\int_{x}^{b}
\exp\left(-\int_{x}^{s}\frac{W_{q}^{\prime}(\overline{\xi}\left(z\right))}
{W_{q}(\overline{\xi}\left(z\right))}\mathrm{d}z\right)
\nonumber\\
\hspace{-0.3cm}&&\hspace{0.3cm}
\times\left(Z_{q}(\overline{\xi}(s))-\psi^{\prime}(0+)W_{q}(\overline{\xi}(s))-\frac{\overline{Z}_{q}(\overline{\xi}(s))-\psi^{\prime}(0+)
\overline{W}_{q}(\overline{\xi}(s))}{W_{q}(\overline{\xi}(s))}W_{q}^{\prime}(\overline{\xi}(s))\right)\mathrm{d}s,\quad x\in(-\infty, b]
.
\end{eqnarray}
\end{lem}

\medskip
We start with the Laplace transform of the upper exiting time for the process $U$.

\medskip
\begin{prop}\label{3.1}
For $q,\lambda\in(0,\infty)$ we have
\begin{eqnarray}
\label{upper.lower.boun.resu.}
\mathbb{E}_{x}\left(\mathrm{e}^{-q\kappa_{b}^{+}}\mathbf{1}_{\{\kappa_{b}^{+}<\theta_{\lambda}\}}\right)
\hspace{-0.3cm}&=&\hspace{-0.3cm}
\exp\left(-\int_{x}^{b}\ell_{1}(w)\mathrm{d}w\right),\quad x\in(-\infty,b],
\end{eqnarray}
where
\begin{eqnarray}
\ell_{1}(w)=\frac{W_{q}^{\prime}(\overline{\xi}(w))}
{W_{q}(\overline{\xi}(w))}
\left(1-\frac{Z_{q}(\overline{\xi}(w))}{Z_{q+\lambda}(\overline{\xi}(w))}\right)
+\frac{qW_{q}(\overline{\xi}(w))}{Z_{q+\lambda}(\overline{\xi}(w))}.\nonumber
\end{eqnarray}
\end{prop}

\begin{proof}[Proof:]\,\,\,
Denote by $f(x)$ the left hand side of (\ref{upper.lower.boun.resu.}). We have
\begin{eqnarray}\label{part.}
f(x)=\mathbb{E}_{x}\left(\mathrm{e}^{-q\kappa_{b}^{+}}\mathbf{1}_{\{\kappa_{b}^{+}<\tau_{\xi}\}}\right)
+\mathbb{E}_{x}\left(\mathrm{e}^{-q\kappa_{b}^{+}}\mathbf{1}_{\{\tau_{\xi}<\kappa_{b}^{+}<\theta_{\lambda}\}}\right),\quad x\in(-\infty,b].
\end{eqnarray}
Note that by definition,  $\tau_{\xi}<\kappa_{b}^{+}$ implies $\overline{X}(\tau_{\xi})<b$ which further implies $T_{1}<\kappa_{b}^{+}$. Hence, taking use of \eqref{two.sid.exit.Y} and \eqref{12}, we get for $x\in(-\infty,b]$
\begin{eqnarray}\label{part2}
\hspace{-0.3cm}&&\hspace{-0.3cm}\mathbb{E}_{x}\left(\mathrm{e}^{-q\kappa_{b}^{+}}\mathbf{1}_{\{\tau_{\xi}<\kappa_{b}^{+}<\theta_{\lambda}\}}\right)
=\mathbb{E}_{x}\left(\mathrm{e}^{-q\kappa_{b}^{+}}\mathbf{1}_{\{\tau_{\xi}<T_{1}<\kappa_{b}^{+}<\theta_{\lambda}\}}
\right)
\nonumber\\
\hspace{-0.3cm}&=&\hspace{-0.3cm}
\mathbb{E}_{x}\left(\mathrm{e}^{-q\tau_{\xi}}\mathbf{1}_{\{\tau_{\xi}<\tau_{b}^{+}\}} \left[\left.\mathbb{E}_{}\left(\mathrm{e}^{-q \sigma^{+}_{z}}\mathbf{1}_{\{\sigma^{+}_{z}<e_{\lambda}\}}\right)\right|_{z=\overline{\xi}(\overline{X}(\tau_{\xi}))}\right]f(\overline{X}(\tau_{\xi}))\right)
\nonumber\\
\hspace{-0.3cm}&=&\hspace{-0.3cm}
\mathbb{E}_{x}\left(\mathrm{e}^{-q\tau_{\xi}}\mathbf{1}_{\{\tau_{\xi}<\tau_{b}^{+}\}} \frac{f(\overline{X}(\tau_{\xi}))}{Z_{q+\lambda}(\overline{\xi}(\overline{X}(\tau_{\xi})))}\right)
\nonumber\\
\hspace{-0.3cm}&=&\hspace{-0.3cm}
\int_{x}^{b}\frac{f(s)}{Z_{q+\lambda}(\overline{\xi}(s))}
\exp\left(-\int_{x}^{s}\frac{W_{q}^{\prime}(\overline{\xi}(z))}{ W_{q}(\overline{\xi}(z))}\mathrm{d}z\right)
\left(\frac{W_{q}^{\prime}(\overline{\xi}(s))}{W_{q}(\overline{\xi}(s))}Z_{q}(\overline{\xi}(s))-qW_{q}(\overline{\xi}(s))\right)\mathrm{d}s.
\end{eqnarray}
Combining (\ref{part1}), (\ref{part.}) and (\ref{part2}), we obtain for $x\in(-\infty, b]$
\begin{eqnarray}\label{f(x)}
f(x)\hspace{-0.3cm}&=&\hspace{-0.3cm}\exp\left(-\int_{x}^{b}\frac{W_{q}^{\prime}(\overline{\xi}\left(z\right))}
{W_{q}(\overline{\xi}\left(z\right))}\mathrm{d}z\right)
\nonumber\\
\hspace{-0.3cm}&&\hspace{-0.5cm}
+\int_{x}^{b}\frac{f(s)}{Z_{q+\lambda}(\overline{\xi}(s))}
\exp\left(-\int_{x}^{s}\frac{W_{q}^{\prime}(\overline{\xi}(z))}{ W_{q}(\overline{\xi}(z))}\mathrm{d}z\right)
\left(\frac{W_{q}^{\prime}(\overline{\xi}(s))}{W_{q}(\overline{\xi}(s))}Z_{q}(\overline{\xi}(s))-qW_{q}(\overline{\xi}(s))\right)\mathrm{d}s.
\end{eqnarray}
Taking derivative on both sides of (\ref{f(x)}) with respect to $x$, we have
\begin{eqnarray}\label{f'(x)}
f^{\prime}(x)\hspace{-0.3cm}&=&\hspace{-0.3cm}
\frac{W_{q}^{\prime}(\overline{\xi}\left(x\right))}
{W_{q}(\overline{\xi}\left(x\right))}
f(x)
-\frac{f(x)}{Z_{q+\lambda}(\overline{\xi}(x))}\left(\frac{W_{q}^{\prime}(\overline{\xi}(x))}{W_{q}(\overline{\xi}(x))}Z_{q}(\overline{\xi}(x))-qW_{q}(\overline{\xi}(x))\right)
\nonumber\\
\hspace{-0.3cm}&=&\hspace{-0.3cm}
\ell_{1}(x)f(x),\quad x\in(-\infty,b].
\end{eqnarray}
Solving (\ref{f'(x)}) we obtain for $x\in(-\infty,b]$
\begin{eqnarray}\label{fwithC}
f(x)
\hspace{-0.3cm}&=&\hspace{-0.3cm}
C\exp\left(-\int_{x}^{b}\ell_{1}(w)\mathrm{d}w\right),
\end{eqnarray}
for some constant $C$.
The boundary condition $f(b)=1$ together with (\ref{fwithC}) yields (\ref{upper.lower.boun.resu.}).
\end{proof}

\begin{prop}
For $q,\lambda\in (0,\infty)$, we have for $x\in(-\infty,b]$,
\begin{eqnarray}\label{eq:main1}
\mathbb{E}_x\Big(\mathrm{e}^{-q\theta_{\lambda}}\mathbf{1}_{\{\theta_{\lambda}<\kappa_b^+\}}\Big)=\int_x^b \exp\Big(-\int_x^y \ell_{1}(w)dw\Big)\overline{\ell}_{1}(y)\mathrm{d}y,
\end{eqnarray}
where $\ell_{1}(x)$ is given in (\ref{upper.lower.boun.resu.}), while the function $\overline{\ell}_{1}(x)$ is defined by
\begin{eqnarray*}
\overline{\ell}_{1}(x)
=
\frac{\lambda}{q+\lambda}\left(1 - \frac{1}{Z_{q+\lambda}(\overline{\xi}(x))} \right)\left(\frac{W_q^{\prime}(\overline{\xi}(x))}{W_q(\overline{\xi}(x))}Z_q(\overline{\xi}(x)) - qW_q(\overline{\xi}(x))\right).
\end{eqnarray*}
\end{prop}
\begin{proof}
Denote by $\overline{f}_{\xi}(x;b)$ the left hand side of \eqref{eq:main1}.
We have
\begin{eqnarray}
\overline{f}_{\xi}(x;b) \hspace{-0.3cm}&=&\hspace{-0.3cm}\mathbb{E}_x\Big( \mathrm{e}^{-q \theta_{\lambda}}\mathbf{1}_{\{\theta_{\lambda}<\kappa_b^+, \tau_{\xi}<\tau_b^+\}} \Big) \nonumber\\
\hspace{-0.3cm}&=&\hspace{-0.3cm} \mathbb{E}_x\Big( \mathrm{e}^{-q \theta_{\lambda}}\mathbf{1}_{\{\theta_{\lambda}<\kappa_b^+, \tau_{\xi}<\tau_b^+, T_1-\tau_{\xi} >  e_{\lambda}\}} \Big) + \mathbb{E}_x\Big( \mathrm{e}^{-q \theta_{\lambda}}\mathbf{1}_{\{\theta_{\lambda}<\kappa_b^+, \tau_{\xi}<\tau_b^+,T_1-\tau_{\xi} \leq  e_{\lambda}\}} \Big). \label{eq:eq1}
\end{eqnarray}
By the strong Markov property of the bi-variate process $(U,\overline{U})$ one can obtain
\begin{eqnarray}
\mathbb{E}\Big( \mathrm{e}^{-q\theta_{\lambda}}\mathbf{1}_{\{\theta_{\lambda}<\kappa_b^+, \tau_{\xi}<\tau_b^+, T_1-\tau_{\xi} >  e_{\lambda}\}}\Big\vert \mathcal{F}_{\tau_{\xi}}\Big)
\hspace{-0.3cm}&=&\hspace{-0.3cm}
\mathrm{e}^{-q\tau_{\xi}}\mathbf{1}_{\{\tau_{\xi}<\tau_b^+\}}
\left[\left.\mathbb{E}_{}\Big(\mathrm{e}^{-q e_{\lambda}}\mathbf{1}_{\{ e_{\lambda}<\sigma_{z}^+\}}\Big)\right|_{z=\overline{\xi}(\overline{X}_{\tau_{\xi}})}\right] \nonumber\\
\hspace{-0.3cm}&=&\hspace{-0.3cm}
\mathrm{e}^{-q\tau_{\xi}}\mathbf{1}_{\{\tau_{\xi}<\tau_b^+\}}\frac{\lambda}{q+\lambda}\Big(1-\Big[\mathbb{E}\Big(\mathrm{e}^{-(q+\lambda)\sigma_z^+}\Big)\Big\vert_{z=\overline{\xi}(\overline{X}_{\tau_{\xi}})}\Big] \Big) \nonumber\\
\hspace{-0.3cm}&=&\hspace{-0.3cm}
\mathrm{e}^{-q\tau_{\xi}}\mathbf{1}_{\{\tau_{\xi}<\tau_b^+\}}\frac{\lambda}{q+\lambda}\Big(1-\frac{1}{Z_{q+\lambda}(\overline{\xi}(\overline{X}_{\tau_{\xi}}))} \Big), \label{eq:eq1a}
\end{eqnarray}
and
\begin{eqnarray}
\hspace{-0.3cm}&&\hspace{-0.3cm}
\mathbb{E}\Big(\mathbb{E}\Big( \mathrm{e}^{-q\theta_{\lambda}}\mathbf{1}_{\{\theta_{\lambda}<\kappa_b^+, \tau_{\xi}<\tau_b^+, T_1-\tau_{\xi} \leq  e_{\lambda}\}}\Big\vert \mathcal{F}_{T_{1}}\Big)\Big\vert \mathcal{F}_{\tau_{\xi}}\Big)
\nonumber\\
\hspace{-0.3cm}&=&\hspace{-0.3cm}
\mathrm{e}^{-q\tau_{\xi}}\mathbf{1}_{\{\tau_{\xi}<\tau_b^+\}}
\left[\left.\mathbb{E}_{}\Big(\mathrm{e}^{-q \sigma_{z}^+}\mathbf{1}_{\{ \sigma_{z}^+\leq e_{\lambda}\}}\Big)\right|_{z=\overline{\xi}(\overline{X}_{\tau_{\xi}})}\right]\mathbb{E}_{\overline{X}_{\tau_{\xi}}}\Big(\mathrm{e}^{-q\theta_{\lambda}}\mathbf{1}_{\{\theta_{\lambda}<\kappa_b^+\}}\Big)
\nonumber\\
\hspace{-0.3cm}&=&\hspace{-0.3cm}
\mathrm{e}^{-q\tau_{\xi}}\mathbf{1}_{\{\tau_{\xi}<\tau_b^+\}}\frac{\overline{f}_{\xi}(\overline{X}_{\tau_{\xi}};b)}{Z_{q+\lambda}(\overline{\xi}(\overline{X}_{\tau_{\xi}}))}\label{eq:eq1b}.
\end{eqnarray}
Following the two identities (\ref{eq:eq1a}) and (\ref{eq:eq1b}), we obtain from (\ref{eq:eq1}) and (\ref{12}) the equation
\begin{eqnarray*}
&&\overline{f}_{\xi}(x;b)=\frac{\lambda}{q+\lambda}\mathbb{E}_x\left(\mathrm{e}^{-q\tau_{\xi}}\mathbf{1}_{\{\tau_{\xi}<\tau_b^+\}}\Big(1-\frac{1}{Z_{q+\lambda}(\overline{\xi}(\overline{X}_{\tau_{\xi}}))} \Big) \right) \nonumber\\
&&\hspace{2cm}+ \mathbb{E}\left(\mathrm{e}^{-q\tau_{\xi}}\mathbf{1}_{\{\tau_{\xi}<\tau_b^+\}}\frac{\overline{f}_{\xi}(\overline{X}_{\tau_{\xi}};b)}{Z_{q+\lambda}(\overline{\xi}(\overline{X}_{\tau_{\xi}}))}\right) \nonumber\\
&&= \frac{\lambda}{q+\lambda} \int_x^b \Big[ 1- \frac{1}{Z_{q+\lambda}(\overline{\xi}(s))} \Big]\exp\left(-\int_x^s \frac{W_q^{\prime}(\overline{\xi}(z))}{W_q(\overline{\xi}(z))}dz\right)\left(\frac{W_q^{\prime}(\overline{\xi}(s))}{W_q(\overline{\xi}(s))} Z_q(\overline{\xi}(s)) - q W_q(\overline{\xi}(s))  \right)\mathrm{d}s\\
&& \hspace{1cm}+ \int_x^b\frac{\overline{f}_{\xi}(s;b)}{Z_{q+\lambda}(\overline{\xi}(s))}\exp\left(-\int_x^s \frac{W_q^{\prime}(\overline{\xi}(z))}{W_q(\overline{\xi}(z))}dz\right)\left(\frac{W_q^{\prime}(\overline{\xi}(s))}{W_q(\overline{\xi}(s))} Z_q(\overline{\xi}(s)) - q W_q(\overline{\xi}(s))  \right)\mathrm{d}s.
\end{eqnarray*}
After taking partial derivative w.r.t $x$ on both sides, we have by the  Leibniz integral rule,
\begin{eqnarray}
\overline{f}_{\xi}^{\prime}(x;b)
\hspace{-0.3cm}&=&\hspace{-0.3cm}
-\frac{\lambda}{q+\lambda}\Big[1-\frac{1}{Z_{q+\lambda}(\overline{\xi}(x))}\Big]
\left(\frac{W_q^{\prime}(\overline{\xi}(x))}{W_q(\overline{\xi}(x))} Z_q(\overline{\xi}(x)) - q W_q(\overline{\xi}(x)) \right) \nonumber\\
\hspace{-0.3cm}&&\hspace{-0.3cm}
-\frac{\overline{f}_{\xi}(x;b)}{Z_{q+\lambda}(\overline{\xi}(x))}\left(\frac{W_q^{\prime}(\overline{\xi}(x))}{W_q(\overline{\xi}(x))} Z_q(\overline{\xi}(x)) - q W_q(\overline{\xi}(x))  \right) + \frac{W_q^{\prime}(\overline{\xi}(x))}{W_q(\overline{\xi}(x))}\overline{f}_{\xi}(x;b) \nonumber\\
\hspace{-0.3cm}&=&\hspace{-0.3cm}
\ell_1(x)\overline{f}_{\xi}(x;b) - \overline{\ell}_{1}(x)\label{eq:eq3}.
\end{eqnarray}
Identity (\ref{eq:main1}) follows by solving the differential equation (\ref{eq:eq3}) subject to the boundary condition $\overline{f}_{\xi}(b;b)=0$. \end{proof}

\medskip

\begin{cor}[Parisian Ruin probability]
For $\lambda\in (0,\infty)$, the Parisian ruin probability is given by
\begin{eqnarray}\label{eq:ruinprobab}
\mathbb{P}_x\big(\theta_{\lambda}<\infty)=1-\exp\Big(-\int_x^{\infty} \ell_1(w;\lambda)dw\Big), \quad x\in\mathbb{R} ,
\end{eqnarray}
with
\begin{eqnarray*}
\ell_1(x;\lambda)=\frac{W^{\prime}(\overline{\xi}(x))}{W(\overline{\xi}(x))}\left(1-\frac{1}{Z_{\lambda}(\overline{\xi}(x))}\right).
\end{eqnarray*}
\end{cor}

\medskip

\begin{thm}
For $q,\lambda\in(0,\infty)$, we have for $x\in(-\infty,b]$ that
\begin{eqnarray}
U_{\xi}(x;b)=\int_x^b \exp\Big(-\int_x^y \ell_{1}(w)dw\Big)\overline{\ell}_{1}(y)\mathrm{d}y +  \exp\Big(-\int_x^b \ell_{1}(w)dw\Big).
\end{eqnarray}
\begin{proof}
The proof follows from combining the two results (\ref{upper.lower.boun.resu.}) and (\ref{eq:main1}).
\end{proof}
\end{thm}

\medskip
The next result gives an expression of the joint Laplace transform concerning $\kappa_{b}^{+}\wedge \theta_{\lambda}$.

\medskip
\begin{thm}\label{3.4}
For any $q, u,v,\lambda\in(0,\infty)$ with $u\in(0,\Phi_{q+\lambda})$, we have
\begin{eqnarray}
\label{lap.of.exp.tot.dis.cap.inj.}
\quad\quad G_\xi(x;b)
\hspace{-0.3cm}&=&\hspace{-0.3cm}
\left(\mathrm{e}^{ub}-\int_{x}^{b}\overline{\ell}_{2}(z)\,\exp\left(\int_{z}^{b}\ell_{2}(w)\mathrm{d}w\right)\mathrm{d}z\right)\exp\left(-\int_{x}^{b}\ell_{2}(w)\mathrm{d}w\right),\quad x\in(-\infty, b],
\end{eqnarray}
where
\begin{eqnarray}
\ell_{2}(w)
\hspace{-0.3cm}&=&\hspace{-0.3cm}
\frac{W_{q}^{\prime}(\overline{\xi}(w))}
{W_{q}(\overline{\xi}(w))}
\left(1-\frac{Z_{q}(\overline{\xi}(w),v)}
{Z_{q+\lambda}(\overline{\xi}(w),v)}\right)+
\frac{v  Z_{q}(\overline{\xi}(w),v)+(q-\psi(v ))W_{q}(\overline{\xi}(w))}
{Z_{q+\lambda}(\overline{\xi}(w),v)},
\nonumber\\
\overline{\ell}_{2}(w)\hspace{-0.3cm}&=&\hspace{-0.3cm}
-\mathrm{e}^{u\xi(w)}
\hbar(\overline{\xi}(w))
\left(\frac{W_{q}^{\prime}(\overline{\xi}(w))}{W_{q}(\overline{\xi}(w))}Z_{q}(\overline{\xi}(w),v)-v  Z_{q}(\overline{\xi}(w),v)-(q-\psi(v ))W_{q}(\overline{\xi}(w))\right)
,\nonumber
\\
\label{hbar}
\hbar(w)
\hspace{-0.3cm}&=&\hspace{-0.3cm}
\frac{\lambda}{\Phi_{q+\lambda}-u}
\left(
W_{q+\lambda}(0+)
+\int_{y<0}\mathrm{e}^{u y+(v-u) y}\left(W_{q+\lambda}^{\prime}(-y)-\Phi_{q+\lambda}W_{q+\lambda}(-y)\right)
\mathrm{d}y\right.
\nonumber\\
\hspace{-0.3cm}&&\hspace{-0.3cm}
-\frac{1}{Z_{q+\lambda}(w,v)}
\bigg(
W_{q+\lambda}(0+)\,\mathrm{e}^{uw}
\nonumber\\
\hspace{-0.3cm}&&\hspace{-0.3cm}
\left.\left.+\int_{y<w}\mathrm{e}^{u y+(v-u) (y\wedge 0)}\left(W_{q+\lambda}^{\prime}(w-y)-\Phi_{q+\lambda}W_{q+\lambda}(w-y)\right)
\mathrm{d}y\right)\right).\nonumber
\end{eqnarray}
\end{thm}

\begin{proof}[Proof:]\,\,\,
Given $q, u,v\in(0,\infty)$ and $x\in(-\infty, b]$, applying the Markov property of the process $(U,\overline{U})$ we have
\begin{eqnarray}\label{add.1}
\hspace{-0.3cm}&&\hspace{-0.3cm}
G_\xi(x;b):=
\mathbb{E}_{x}\left(\mathrm{e}^{-q \left(\kappa_{b}^{+}\wedge \theta_{\lambda}\right)+u U(\kappa_{b}^{+}\wedge \theta_{\lambda})-v R\left(\kappa_{b}^{+}\wedge \theta_{\lambda}\right)}\right)
\nonumber\\
\hspace{-0.3cm}&=&\hspace{-0.3cm}
\mathbb{E}_{x}\left(\mathbb{E}_{x}\left(\left.\mathrm{e}^{-q \,(\tau_{\xi}+e_{\lambda}^{(1)})+u U(\tau_{\xi}+e_{\lambda}^{(1)})-v R(\tau_{\xi}+e_{\lambda}^{(1)})}\mathbf{1}_{\{\tau_{\xi}<\kappa_{b}^{+},\,T_{1}-\tau_{\xi}>e_{\lambda}^{(1)}\}}\right|\mathcal{F}_{\tau_{\xi}}
\right)\right)
\nonumber\\
\hspace{-0.3cm}&&\hspace{-0.3cm}
+
\mathbb{E}_{x}\left(\mathbb{E}_{x}\left(\left.\mathrm{e}^{-q \left(\kappa_{b}^{+}\wedge \theta_{\lambda}\right)+u U\left(\kappa_{b}^{+}\wedge \theta_{\lambda}\right)-v R\left(\kappa_{b}^{+}\wedge \theta_{\lambda}\right)}\mathbf{1}_{\{\tau_{\xi}<\kappa_{b}^{+},\,T_{1}-\tau_{\xi}\leq e_{\lambda}^{(1)}\}}\right|\mathcal{F}_{T_{1}}
\right)\right)
\nonumber\\
\hspace{-0.3cm}&&\hspace{-0.3cm}
+
\mathbb{E}_{x}\left(\mathrm{e}^{-q \tau_{b}^{+}+u X(\tau_{b}^{+})-v R(\tau_{b}^{+})}\mathbf{1}_{\{\tau_{b}^{+}<\tau_{\xi}\}}\right)
\nonumber\\
\hspace{-0.3cm}&=&\hspace{-0.3cm}
\mathbb{E}_{x}\left(
\mathrm{e}^{-q \tau_{\xi}+(u-v)\xi(\overline{X}(\tau_{\xi}))+vX(\tau_{\xi})}\mathbf{1}_{\{\tau_{\xi}<\tau_{b}^{+}\}}
\left(\left.\mathbb{E}_{}\left(\mathrm{e}^{-q e_{\lambda}+u Y(e_{\lambda})+v(\underline{X}(e_{\lambda}))}\mathbf{1}_{\{\sigma_{z}^{+}>e_{\lambda}\}}
\right)\right|_{z=\overline{\xi}(\overline{X}(\tau_{\xi}))}\right)\right)
\nonumber\\
\hspace{-0.3cm}&&\hspace{-0.3cm}
+
\mathbb{E}_{x}\left(\mathrm{e}^{-q \tau_{\xi}}\mathbf{1}_{\{\tau_{\xi}<\kappa_{b}^{+}\}}\mathrm{e}^{-v(\xi(\overline{X}(\tau_{\xi}))-X(\tau_{\xi}))}
\,G_{\xi}(\overline{X}(\tau_{\xi});b)\,\mathbb{E}_{x}\left(\left.
\mathrm{e}^{-v (R(T_{1})-R(\tau_{\xi}))-(q+\lambda) (T_{1}-\tau_{\xi})}\right|\mathcal{F}_{\tau_{\xi}}\right)\right)
\nonumber\\
\hspace{-0.3cm}&&\hspace{-0.3cm}
+
\mathbb{E}_{x}\left(\mathrm{e}^{-q \tau_{b}^{+}+u X(\tau_{b}^{+})-v R(\tau_{b}^{+})}\mathbf{1}_{\{\tau_{b}^{+}<\tau_{\xi}\}}\right)
\nonumber\\
\hspace{-0.3cm}&=&\hspace{-0.3cm}
\mathbb{E}_{x}\left(
\mathrm{e}^{-q \tau_{\xi}+(u-v)\xi(\overline{X}(\tau_{\xi}))+vX(\tau_{\xi})}\mathbf{1}_{\{\tau_{\xi}<\tau_{b}^{+}\}}
\left(\left.\mathbb{E}_{}\left(\mathrm{e}^{-q e_{\lambda}+u Y(e_{\lambda})+v(\underline{X}(e_{\lambda}))}\mathbf{1}_{\{\sigma_{z}^{+}>e_{\lambda}\}}
\right)\right|_{z=\overline{\xi}(\overline{X}(\tau_{\xi}))}\right)\right)
\nonumber\\
\hspace{-0.3cm}&&\hspace{-0.3cm}
+
\mathbb{E}_{x}\left(\mathrm{e}^{-q \tau_{\xi}}\mathbf{1}_{\{\tau_{\xi}<\kappa_{b}^{+}\}}\mathrm{e}^{vX(\tau_{\xi})}
\mathrm{e}^{-v\xi(\overline{X}(\tau_{\xi}))}
\,G_{\xi}(\overline{X}(\tau_{\xi});b)\left(\left.\mathbb{E}_{}\left(
\mathrm{e}^{v(\underline{X}(\sigma_{z}^{+}))-(q+\lambda) \sigma_{z}^{+}}\right)\right|_{z=\overline{\xi}(\overline{X}(\tau_{\xi}))}\right)\right)
\nonumber\\
\hspace{-0.3cm}&&\hspace{-0.3cm}
+
\mathrm{e}^{u b}\,\mathbb{E}_{x}\left(\mathrm{e}^{-q \tau_{b}^{+}}\mathbf{1}_{\{\tau_{b}^{+}<\tau_{\xi}\}}\right).
\end{eqnarray}
In addition,
\begin{eqnarray}
\label{w36}
\hbar(z)\hspace{-0.3cm}&:=&\hspace{-0.3cm}\mathbb{E}_{}\left(\mathrm{e}^{-q e_{\lambda}+u Y(e_{\lambda})+v(\underline{X}(e_{\lambda})\wedge 0)}\mathbf{1}_{\{\sigma_{z}^{+}>e_{\lambda}\}}
\right)
\nonumber\\
\hspace{-0.3cm}&=&\hspace{-0.3cm}\mathbb{E}_{}\left(\mathrm{e}^{-q e_{\lambda}+u Y(e_{\lambda})+v(\underline{X}(e_{\lambda})\wedge 0)}\right)
-\mathbb{E}_{}\left(\mathrm{e}^{-q e_{\lambda}+u Y(e_{\lambda})+v(\underline{X}(e_{\lambda})\wedge 0)}\mathbf{1}_{\{\sigma_{z}^{+}<e_{\lambda}\}}\right)
\nonumber\\
\hspace{-0.3cm}&=&\hspace{-0.3cm}
\lambda\,\mathbb{E}_{}\left(
\int_{0}^{\infty}\mathrm{e}^{-(q+\lambda) t}\,\mathrm{e}^{u Y(t)+v(\underline{X}(t)\wedge 0)}
\mathrm{d}t
-\int_{\sigma_{z}^{+}}^{\infty}\mathrm{e}^{-(q+\lambda) t}\,\mathrm{e}^{u Y(t)+v(\underline{X}(t)\wedge 0)}
\mathrm{d}t\right)
\nonumber\\
\hspace{-0.3cm}&=&\hspace{-0.3cm}
\frac{\lambda}{q+\lambda}\mathbb{E}_{}\left(
\mathrm{e}^{u Y(e_{q+\lambda})+v(\underline{X}(e_{q+\lambda})\wedge 0)}\right)
\nonumber\\
\hspace{-0.3cm}&&\hspace{-0.3cm}
-
\lambda\,\mathbb{E}_{}\left(\int_{0}^{\infty}
\mathbb{E}_{}\left(\left.\mathrm{e}^{-(q+\lambda) (\sigma_{z}^{+}+t)}\,\mathrm{e}^{u Y(\sigma_{z}^{+}+t)+v(\underline{X}(\sigma_{z}^{+}+t)\wedge 0)}
\right|\mathcal{F}_{\sigma_{z}^{+}}\right)\mathrm{d}t\right)
\nonumber\\
\hspace{-0.3cm}&=&\hspace{-0.3cm}
\frac{\lambda}{q+\lambda}\mathbb{E}_{}\left(
\mathrm{e}^{u Y(e_{q+\lambda})+v(\underline{X}(e_{q+\lambda})\wedge 0)}\right)
\nonumber\\
\hspace{-0.3cm}&&\hspace{-0.3cm}
-\frac{\lambda}{q+\lambda}\mathbb{E}_{}\left(
\mathrm{e}^{-(q+\lambda) \sigma_{z}^{+}+v(\underline{X}(\sigma_{z}^{+}))}
\,\mathbb{E}_{z}\left(\mathrm{e}^{u Y(e_{q+\lambda})+v(\underline{X}(e_{q+\lambda})\wedge 0)}
\right)\right)
\nonumber\\
\hspace{-0.3cm}&=&\hspace{-0.3cm}
\frac{\lambda}{q+\lambda}\mathbb{E}_{}\left(
\mathrm{e}^{u X(e_{q+\lambda})+(v-u)(\underline{X}(e_{q+\lambda})\wedge 0)}\right)
\nonumber\\
\hspace{-0.3cm}&&\hspace{-0.3cm}
-
\frac{\lambda}{q+\lambda}\,\mathbb{E}_{}\left(
\mathrm{e}^{-(q+\lambda) \sigma_{z}^{+}+v(\underline{X}(\sigma_{z}^{+}))}
\right)
\,\mathbb{E}_{z}\left(\mathrm{e}^{u X(e_{q+\lambda})+(v-u) (\underline{X}(e_{q+\lambda})\wedge 0)}
\right),\quad z\in(0,\infty),
\end{eqnarray}
where for the fifth equality we have taken use of
\begin{eqnarray}
\mathbb{E}_{}\left(\left.\mathrm{e}^{u Y(\sigma_{z}^{+}+t)+v(\underline{X}(\sigma_{z}^{+}+t)\wedge 0)}
\right|\mathcal{F}_{\sigma_{z}^{+}}\right)
\hspace{-0.3cm}&=&\hspace{-0.3cm}
\mathbb{E}_{}\left(\left.\mathrm{e}^{(u-v) Y(\sigma_{z}^{+}+t)+v(X(\sigma_{z}^{+}+t)-X(\sigma_{z}^{+})+z+X(\sigma_{z}^{+})-z)}
\right|\mathcal{F}_{\sigma_{z}^{+}}\right)
\nonumber\\
\hspace{-0.3cm}&=&\hspace{-0.3cm}
\mathrm{e}^{v(X(\sigma_{z}^{+})-z)}
\mathbb{E}_{}\left(\left.\mathrm{e}^{(u-v) Y(\sigma_{z}^{+}+t)+v(X(\sigma_{z}^{+}+t)-X(\sigma_{z}^{+})+z)}
\right|\mathcal{F}_{\sigma_{z}^{+}}\right)
\nonumber\\
\hspace{-0.3cm}&=&\hspace{-0.3cm}
\mathrm{e}^{v(Y(\sigma_{z}^{+})+(\underline{X}(\sigma_{z}^{+})\wedge 0)-z)}
\mathbb{E}_{z}\left(\mathrm{e}^{(u-v) Y(t)+vX(t)}\right)
\nonumber\\
\hspace{-0.3cm}&=&\hspace{-0.3cm}
\mathrm{e}^{v(\underline{X}(\sigma_{z}^{+})\wedge 0)}
\mathbb{E}_{z}\left(\mathrm{e}^{u Y(t)+v(\underline{X}(t)\wedge 0)}\right),
\nonumber
\end{eqnarray}
which holds true since
\begin{eqnarray}
Y(\sigma_{z}^{+}+t)
\hspace{-0.3cm}&=&\hspace{-0.3cm}
X(\sigma_{z}^{+}+t)-X(\sigma_{z}^{+})-\inf_{s\leq t}(X(\sigma_{z}^{+}+s)-X(\sigma_{z}^{+}))
\wedge (-Y(\sigma_{z}^{+}))
\nonumber\\
\hspace{-0.3cm}&=&\hspace{-0.3cm}
X(\sigma_{z}^{+}+t)-X(\sigma_{z}^{+})+z-\inf_{s\leq t}(X(\sigma_{z}^{+}+s)-X(\sigma_{z}^{+})+z)
\wedge 0
\nonumber\\
\hspace{-0.3cm}&:=&\hspace{-0.3cm}
\widetilde{X}(t)-\underline{\widetilde{X}}(t)\wedge 0,
\nonumber
\end{eqnarray}
where $\underline{\widetilde{X}}$ denotes the running infimum process of the process $\widetilde{X}=\{X(\sigma_{z}^{+}+t)-X(\sigma_{z}^{+})+z; t\geq 0\}$ which starts from $z\in(0,\infty)$, is independent of $\mathcal{F}_{\sigma_{z}^{+}}$ and is identical in law to $(X, \mathbb{P}_{z})$.

By adapting (24) in Albrecher et al. (2016) one has
\begin{eqnarray}
\label{w37}
\mathbb{E}_{}\left(
\mathrm{e}^{-(q+\lambda) \sigma_{z}^{+}+v(\underline{X}(\sigma_{z}^{+}))}
\right)
\hspace{-0.3cm}&=&\hspace{-0.3cm}
\mathbb{E}_{}\left(
\mathrm{e}^{v(\underline{X}(\sigma_{z}^{+}))}
\mathbf{1}_{\{\sigma_{z}^{+}<e_{q+\lambda}\}}
\right)
\nonumber\\
\hspace{-0.3cm}&=&\hspace{-0.3cm}
\frac{1}{Z_{q+\lambda}(z,v)},\quad z\in(0,\infty).
\end{eqnarray}
In addition, by Lemma 1 of Bertoin (1997) with minor adaptation one has
\begin{eqnarray}
\hspace{-0.3cm}&&\hspace{-0.3cm}
\mathbb{P}_{z}\left(X(e_{q+\lambda})\in \mathrm{d}x,
\underline{X}(e_{q+\lambda})\geq y\right)
\nonumber\\
\hspace{-0.3cm}&=&\hspace{-0.3cm}
(q+\lambda)\left(\mathrm{e}^{-\Phi_{q+\lambda}(x-y)}W_{q+\lambda}(z-y)-\mathbf{1}_{\{z\geq x\}}W_{q+\lambda}(z-x)\right)\mathrm{d}x,\quad y\in(-\infty, x\wedge z],\, z\in(0,\infty),\nonumber
\end{eqnarray}
which implies that
\begin{eqnarray}
\label{w38}
\hspace{-0.3cm}&&\hspace{-0.3cm}
\mathbb{P}_{z}\left(X(e_{q+\lambda})\in \mathrm{d}x,
\underline{X}(e_{q+\lambda})\in \mathrm{d}y\right)
\nonumber\\
\hspace{-0.3cm}&=&\hspace{-0.3cm}
(q+\lambda)\,\mathrm{e}^{-\Phi_{q+\lambda}(x-y)}\left(W_{q+\lambda}^{\prime}(z-y)-\Phi_{q+\lambda}W_{q+\lambda}(z-y)\right)
\mathbf{1}_{\{y<z\}}\mathbf{1}_{\{y\leq x\}}\mathrm{d}x\mathrm{d}y
\nonumber\\
\hspace{-0.3cm}&&\hspace{-0.3cm}
+(q+\lambda)\,\mathrm{e}^{-\Phi_{q+\lambda}(x-z)}
W_{q+\lambda}(0+)\mathbf{1}_{\{z\leq x\}}\delta_{z}(\mathrm{d}y) \mathrm{d}x,\quad y\in(-\infty, z],\,x\in [y,\infty),\, z\in(0,\infty).
\end{eqnarray}
Combining \eqref{w36}, \eqref{w37} and \eqref{w38}, we have for $z\in(0,\infty)$
\begin{eqnarray}
\hbar(z)
\hspace{-0.3cm}&=&\hspace{-0.3cm}
\lambda
\left(W_{q+\lambda}(0+)\int_{0}^{\infty}
\mathrm{e}^{u x}
\mathrm{e}^{-\Phi_{q+\lambda}x} \mathrm{d}x
\right.
\nonumber\\
\hspace{-0.3cm}&&\hspace{-0.3cm}
\left.+\int_{y<0}\int_{x\geq y}\mathrm{e}^{u x+(v-u) (y\wedge 0)}\mathrm{e}^{-\Phi_{q+\lambda}(x-y)}\left(W_{q+\lambda}^{\prime}(-y)-\Phi_{q+\lambda}W_{q+\lambda}(-y)\right)
\mathrm{d}x\mathrm{d}y\right)
\nonumber\\
\hspace{-0.3cm}&&\hspace{-0.3cm}
-\frac{\lambda}{Z_{q+\lambda}(z,v)}
\left(W_{q+\lambda}(0+)\int_{z}^{\infty}
\mathrm{e}^{u x}
\mathrm{e}^{-\Phi_{q+\lambda}(x-z)} \mathrm{d}x
\right.
\nonumber\\
\hspace{-0.3cm}&&\hspace{-0.3cm}
\left.+\int_{y<z}\int_{x\geq y}\mathrm{e}^{u x+(v-u) (y\wedge 0)}\mathrm{e}^{-\Phi_{q+\lambda}(x-y)}\left(W_{q+\lambda}^{\prime}(z-y)-\Phi_{q+\lambda}W_{q+\lambda}(z-y)\right)
\mathrm{d}x\mathrm{d}y\right)
.\nonumber
\end{eqnarray}
By \eqref{10}, \eqref{w36} and \eqref{w37} one can rewrite \eqref{add.1} as
\begin{eqnarray}
\label{G.int.}
G_\xi(x;b)
\hspace{-0.3cm}&=&\hspace{-0.3cm}
\mathbb{E}_{x}\left(
\mathrm{e}^{-q \tau_{\xi}+vX(\tau_{\xi})}
\mathrm{e}^{(u-v)\xi(\overline{X}(\tau_{\xi}))}
\hbar(\overline{\xi}(\overline{X}(\tau_{\xi})))
\mathbf{1}_{\{\tau_{\xi}<\tau_{b}^{+}\}}\right)
\nonumber\\
\hspace{-0.3cm}&&\hspace{-0.3cm}
+
\mathbb{E}_{x}\left(\mathrm{e}^{-q \tau_{\xi}+vX(\tau_{\xi})}
\frac{\mathrm{e}^{-v\xi(\overline{X}(\tau_{\xi}))}G_{\xi}(\overline{X}(\tau_{\xi});b)}{Z_{q+\lambda}(\overline{\xi}(\overline{X}(\tau_{\xi})),v)}
\mathbf{1}_{\{\tau_{\xi}<\kappa_{b}^{+}\}}
\right)
\nonumber\\
\hspace{-0.3cm}&&\hspace{-0.3cm}
+
\mathrm{e}^{u b}\,\mathbb{E}_{x}\left(\mathrm{e}^{-q \tau_{b}^{+}}\mathbf{1}_{\{\tau_{b}^{+}<\tau_{\xi}\}}\right)
\nonumber\\
\hspace{-0.3cm}&=&\hspace{-0.3cm}
\int_{x}^{b}\left(\mathrm{e}^{u\xi(s)}
\hbar(\overline{\xi}(s))+
\frac{G_{\xi}(s;b)}{Z_{q+\lambda}(\overline{\xi}(s),v)}
\right)
\exp\left(-\int_{x}^{s}\frac{W_{q}^{\prime}(\overline{\xi}\left(z\right))}
{W_{q}(\overline{\xi}\left(z\right))}\mathrm{d}z\right)
\nonumber\\
\hspace{-0.3cm}&&\hspace{-0.3cm}
\times\left(\frac{W_{q}^{\prime}(\overline{\xi}(s))}{W_{q}(\overline{\xi}(s))}Z_{q}(\overline{\xi}(s),v)-v  Z_{q}(\overline{\xi}(s),v)-(q-\psi(v ))W_{q}(\overline{\xi}(s))\right)
\mathrm{d}s
\nonumber\\
\hspace{-0.3cm}&&\hspace{-0.3cm}
+\mathrm{e}^{ub}\exp\left(-\int_{x}^{b}\left(\frac{W_{q}^{\prime}(\overline{\xi}\left(z\right))}
{W_{q}(\overline{\xi}\left(z\right))}
\right)\mathrm{d}z\right).
\end{eqnarray}
Differentiating \eqref{G.int.} with respect to $x$ gives
\begin{eqnarray}
\label{G.dif.}
G_{\xi}^{\prime}(x,b)
\hspace{-0.3cm}&=&\hspace{-0.3cm}
\ell_{2}(x)G_{\xi}(x,b)+\overline{\ell}_{2}(x).
\end{eqnarray}
Solving (\ref{G.dif.}) with boundary condition $G_{\xi}(b;b)=\mathrm{e}^{ub}$, we obtain (\ref{lap.of.exp.tot.dis.cap.inj.}).
\end{proof}

\medskip
We then obtain an expression of the resolvent density for the process $U$.

\medskip
\begin{thm}\label{3.2}
\label{reso.meas.U.}
For $q\in (0,\infty)$, the resolvent measure of $U$ is absolutely continuous with respect to the Lebesgue measure with  density given by
\begin{eqnarray}
\label{resovent.meas.}
\hspace{-0.3cm}&&\hspace{-0.3cm}
\int_{0}^{\infty}\mathrm{e}^{-qt}\mathbb{P}_{x}\left(U(t)\in \mathrm{d}u,t<\kappa_{b}^{+}\wedge \theta_{\lambda}\right)\mathrm{d}t
\nonumber\\
\hspace{-0.3cm}&=&\hspace{-0.3cm}
W_{q}(0)
\exp\left(-\int_{x}^{y}
\ell_{1}(w)
\mathrm{d}w\right)
\mathbf{1}_{(x,b)}(u)\mathrm{d}u
\nonumber\\
\hspace{-0.3cm}&&\hspace{-0.3cm}
+\int_{x}^{b}
\exp\left(-\int_{x}^{y}
\ell_{1}(w)
\mathrm{d}w\right)
\ell_{3}(y,u)
\mathbf{1}_{(\xi(y),y)}(u)\mathrm{d}y  \mathrm{d}u,
\quad x,u\in(-\infty, b],
\end{eqnarray}
where $\ell_{1}$ is defined as in Theorem \ref{3.1}, and
\begin{eqnarray}
\ell_{3}(y,u)\hspace{-0.3cm}&=&\hspace{-0.3cm}\left(\frac{W_{q}^{\prime}(\overline{\xi}(y))}{W_{q}(\overline{\xi}(y))}Z_{q}(\overline{\xi}(y))-qW_{q}(\overline{\xi}(y))\right)
\frac{Z_{\lambda}(u-\xi(y))
W_{q}(y-u)}
{Z_{\lambda}(\overline{\xi}(y))Z_{q+\lambda}(\overline{\xi}(y))}
\nonumber\\
\hspace{-0.3cm}&&\hspace{-0.3cm}
+
W_{q}^{\prime}(y-u)-\frac{W_{q}^{\prime}(\overline{\xi}(y))}{W_{q}(\overline{\xi}(y))}W_{q}(y-u).\nonumber
\end{eqnarray}
\end{thm}

\begin{proof}[Proof:]\,\,\,
Recall $\overline{U}(t)=\sup_{s\in[0,t]}U(s)$ and let $e_{q}$ be an exponential  random variable independent of $X$.
For $q>0$, $x\leq b$ and any continuous, non-negative and bounded function $h$, let
\begin{eqnarray}
\label{gene.reso.meas.}
\hspace{-0.3cm}&&\hspace{-0.3cm}qg(x):=\int_{0}^{\infty}q \mathrm{e}^{-qt}\mathbb{E}_{x}\left(h(U(t));t<\kappa_{b}^{+}\wedge \theta_{\lambda}\right)\mathrm{d}t
\nonumber\\
\hspace{-0.3cm}&=&\hspace{-0.3cm}
\mathbb{E}_{x}\left(h(X(e_{q}))\mathbf{1}_{\{X(e_{q})<\overline{X}(e_{q}),\,e_{q}<\tau_{b}^{+}\wedge \tau_{\xi}\}}\right)
+\mathbb{E}_{x}\left(h(U(e_{q}))\mathbf{1}_{\{U(e_{q})<\overline{U}(e_{q}),\tau_{\xi}<e_{q}<\kappa_{b}^{+}\wedge \theta_{\lambda}\}}\right)
\nonumber\\
\hspace{-0.3cm}&&\hspace{-0.3cm}
+\mathbb{E}_{x}\left(\int_{0}^{\infty}q \mathrm{e}^{-qt}h(X(t))\mathbf{1}_{\{X(t)=\overline{X}(t),\,t<\tau_{b}^{+}\wedge \tau_{\xi}\}}\mathrm{d}t\right)
+\mathbb{E}_{x}\left(h(U(e_{q}))\mathbf{1}_{\{U(e_{q})=\overline{U}(e_{q}),\tau_{\xi}<e_{q}<\kappa_{b}^{+}\wedge \theta_{\lambda}\}}\right)
\nonumber\\
\hspace{-0.3cm}&=:&\hspace{-0.3cm}\,
qg_{1}(x)+qg_{2}(x)+qg_{3}(x)+qg_{4}(x).\nonumber
\end{eqnarray}

Note that $\int_{0}^{t}\mathbf{1}_{\{X(s)=\overline{X}(s)\}}\mathrm{d}s=W_{q}(0) \,\overline{X}(t)$ under $\mathbb{P}_{0}$; see Chapters IV and VII of Bertoin (1996), the proof of Part (ii) of Theorem 1 in Pistorius (2004) or the first three paragraphs in Section 5 of Li et al. (2019).
By \eqref{part1} we have
\begin{eqnarray}\label{qg3}
qg_{3}(x)\hspace{-0.3cm}&=&\hspace{-0.3cm}\mathbb{E}_{x}\left(\int_{0}^{\infty}q \mathrm{e}^{-qL^{-1}(L(t))}h(X(L^{-1}(L(t))))\mathbf{1}_{\{X(t)=\overline{X}(t),\,L^{-1}(L(t))<\tau_{b}^{+}\wedge \tau_{\xi}\}}\mathrm{d}t\right)
\nonumber\\
\hspace{-0.3cm}&=&\hspace{-0.3cm}
W(0)\mathbb{E}_{x}\left(\int_{0}^{\infty}q \mathrm{e}^{-qL^{-1}(L(t))}h(X(L^{-1}(L(t))))\mathbf{1}_{\{L^{-1}(L(t))< \tau_{b}^{+}\wedge \tau_{\xi}\}}\mathrm{d}L_{t}\right)
\nonumber\\
\hspace{-0.3cm}&=&\hspace{-0.3cm}
q W(0)\int_{0}^{b-x} \mathbb{E}_{x}\left(\mathrm{e}^{-qL^{-1}(t)}\mathbf{1}_{\{L^{-1}(t)< \tau_{\xi}\}}\right)h(x+t)\mathrm{d}t
\nonumber\\
\hspace{-0.3cm}&=&\hspace{-0.3cm}
q W(0)\int_{x}^{b} \exp\left(-\int_{x}^{s}\frac{W_{q}^{\prime}(\overline{\xi}\left(z\right))}
{W_{q}(\overline{\xi}\left(z\right))}\mathrm{d}z\right)h(s)\mathrm{d}s,
\end{eqnarray}
where we have used the fact that $L^{-1}(t)$ has the same law as the first exit time $\tau_{x+t}^{+}$ under $\mathbb{P}_{x}$.

By the strong Markov property of $(U,\overline{U})$, the definition of $T_{1}=\inf\{t\geq 0: U(t)>\overline{X}(\tau_{\xi}) \}$ (i.e., $\tau_{\xi}<T_{1}$ holds implicitly) given in the definition of $U$ in Section \ref{2}, the memoryless property of the exponentially distributed random variable, as well as \eqref{two.sid.exit.Y} and \eqref{12}, one has
\begin{eqnarray}
\label{25}
\hspace{-0.2cm}
qg_{4}(x)
\hspace{-0.3cm}&=&\hspace{-0.3cm}
\mathbb{E}_{x}\left(\left.\mathbb{E}_{x}\left(h(U(e_{q}))\mathbf{1}_{\{U(e_{q})=\overline{U}(e_{q}),\tau_{\xi}<T_{1}\leq e_{q}<\kappa_{b}^{+}\wedge \theta_{\lambda}\}}\right|\mathcal{F}_{T_{1}}\right)\right)
\nonumber\\
\hspace{-0.3cm}&=&\hspace{-0.3cm}
\mathbb{E}_{x}\left(\mathbf{1}_{\{T_{1}<e_{q}\wedge \kappa_{b}^{+}\wedge \theta_{\lambda}\}}
\mathbb{E}_{\overline{X}(\tau_{\xi})}\left(h(U(e_{q}))\mathbf{1}_{\{U(e_{q})=\overline{U}(e_{q}),e_{q}<\kappa_{b}^{+}\wedge \theta_{\lambda}\}}\right)\right)
\nonumber\\
\hspace{-0.3cm}&=&\hspace{-0.3cm}
\mathbb{E}_{x}\left(\mathbb{E}_{x}\left(\left.\mathrm{e}^{-q T_{1}}\mathrm{e}^{-\lambda (T_{1}-\tau_{\xi})}\mathbf{1}_{\{T_{1}< \kappa_{b}^{+}\}}\left(qg_{3}(\overline{X}(\tau_{\xi}))+qg_{4}(\overline{X}(\tau_{\xi}))\right)\right|\mathcal{F}_{\tau_{\xi}}\right)\right)
\nonumber\\
\hspace{-0.3cm}&=&\hspace{-0.3cm}
q\mathbb{E}_{x}\left(\mathrm{e}^{-q\tau_{\xi}}\mathbf{1}_{\{\tau_{\xi}<\kappa_{b}^{+}\}}\frac{1}{Z_{q+\lambda}(\overline{\xi}(\overline{X}(\tau_{\xi})))}\left(g_{3}(\overline{X}(\tau_{\xi}))+g_{4}(\overline{X}(\tau_{\xi}))\right)\right)
\nonumber\\
\hspace{-0.3cm}&=&\hspace{-0.3cm}
q\int_{x}^{b}\frac{g_{3}(s)+g_{4}(s)}{Z_{q+\lambda}(\overline{\xi}(s))}
\exp\left(-\int_{x}^{s}\frac{W_{q}^{\prime}(\overline{\xi}(z))}{ W_{q}(\overline{\xi}(z))}\mathrm{d}z\right)
\nonumber\\
\hspace{-0.3cm}&&\hspace{-0.3cm}
\times\left(\frac{W_{q}^{\prime}(\overline{\xi}(s))}{W_{q}(\overline{\xi}(s))}Z_{q}(\overline{\xi}(s))-qW_{q}(\overline{\xi}(s))\right)\mathrm{d}s,
\end{eqnarray}
where we also used the fact that
$\tau_{\xi}<\kappa_{b}^{+}$ implies $T_{1}<\kappa_{b}^{+}$
(see also \eqref{part2}), and the fact that
$\tau_{\xi}<e_{q}$
combined with
$U(e_{q})=\overline{U}(e_{q})$ implies $T_{1}\leq e_{q}$.

By the compensation formula, the memoryless property for exponential random variable and \eqref{upper.lower.boun.resu.}, $qg_{1}(x)$ can be expressed as
\begin{eqnarray}\label{qg1.raw.}
\hspace{-0.3cm}&&\hspace{-0.3cm}
\mathbb{E}_{x}\left(\int_{0}^{\infty}\sum_{g}\mathrm{e}^{-qg}\prod\limits_{r<g}\mathbf{1}_{\{\overline{\varepsilon}_{r}\leq \overline{\xi}(x+L(r)),\,L(g)\leq b-x\}}\,h\left(x+L(g)-\varepsilon_{g}(t-g)\right)
\right.
\nonumber\\
\hspace{-0.3cm}&&\hspace{0.5cm}
\left.
\times q\mathrm{e}^{-q (t-g)}\mathbf{1}_{\{g<t<g+\zeta_{g}\wedge \rho_{\overline{\xi}(x+L(g))}^{+}(g)\}}\mathrm{d}t\right)
\nonumber\\
\hspace{-0.3cm}&=&\hspace{-0.3cm}
\mathbb{E}_{x}\left(\sum_{g}\mathrm{e}^{-qg}
\prod\limits_{r<g}\mathbf{1}_{\{\overline{\varepsilon}_{r}\leq \overline{\xi}(x+L(r)),\,L(g)\leq b-x\}}\right.
\nonumber\\
\hspace{-0.3cm}&&\hspace{0.5cm}
\left.
\times\int_{0}^{\infty}q \mathrm{e}^{-q s}h\left(x+L(g)-\varepsilon_{g}(s)\right)\mathbf{1}_{\{s<\zeta_{g}\wedge\rho_{\overline{\xi}(x+L(g))}^{+}(g)\}}
\mathrm{d}s\right)
\nonumber\\
\hspace{-0.3cm}&=&\hspace{-0.3cm}
\mathbb{E}_{x}\left(\int_{0}^{\infty}\mathrm{e}^{-qt}\prod\limits_{r<t}\mathbf{1}_{\{\overline{\varepsilon}_{r}\leq \overline{\xi}(x+L(r)),\,L(t)\leq b-x\}}\right.
\nonumber\\
\hspace{-0.3cm}&&\hspace{0.5cm}
\times
\left.\left(\int_{\mathcal{E}}
\int_{0}^{\infty}q \mathrm{e}^{-q s}h\left(x+L(t)-\varepsilon(s)\right)\mathbf{1}_{\{s<\zeta\wedge\rho_{\overline{\xi}(x+L(t))}^{+}\}}
\mathrm{d}s \,n\left(\mathrm{d}\varepsilon\right)\right)\mathrm{d}L(t)\right)
\nonumber\\
\hspace{-0.3cm}&=&\hspace{-0.3cm}
q\int_{0}^{b-x}\mathbb{E}_{x}\left(\mathrm{e}^{-qL_{t-}^{-1}}
\mathbf{1}_{\{L_{t-}^{-1}< \tau_{\xi}\}}\right)
\int_{0}^{\infty}n\left(\mathrm{e}^{-qs}h(x+t-\varepsilon(s))\mathbf{1}_{\{s<\zeta\wedge\rho^{+}_{\overline{\xi}(x+t)}\}}\right)
\mathrm{d}s\mathrm{d}t
\nonumber\\
\hspace{-0.3cm}&=&\hspace{-0.3cm}
q\int_{x}^{b}\exp\left(-\int_{x}^{t}\frac{W_{q}^{\prime}(\overline{\xi}(z))}{ W_{q}(\overline{\xi}(z))}\mathrm{d}z\right)\int_{0}^{\infty}
n\left(\mathrm{e}^{-qs}h(t-\varepsilon(s))\mathbf{1}_{\{s<\zeta\wedge\rho^{+}_{\overline{\xi}(t)}\}}\right)\mathrm{d}s\mathrm{d}t,
\end{eqnarray}
where $g$ is the left-end point of the excursion $\varepsilon_{g}$, as introduced at the end  of Section 2.
Applying the same arguments as in (\ref{qg3}) and (\ref{qg1.raw.}) we have
\begin{eqnarray}\label{h1}
\hspace{-0.3cm}&&\hspace{-0.3cm}\mathbb{E}_{x}\left(h(X(e_{q}))\mathbf{1}_{\{e_{q}<\tau_{b}^{+}\wedge \tau_{c}^{-}\}}\right)
\nonumber\\
\hspace{-0.3cm}&=&\hspace{-0.3cm}
\mathbb{E}_{x}\left(h(X(e_{q}))\mathbf{1}_{\{X(e_{q})=\overline{X}(e_{q}),\,e_{q}<\tau_{b}^{+}\wedge \tau_{c}^{-}\}}\right)+\mathbb{E}_{x}\left(h(X(e_{q}))\mathbf{1}_{\{X(e_{q})<\overline{X}(e_{q}),\,e_{q}<\tau_{b}^{+}\wedge \tau_{c}^{-}\}}\right)
\nonumber\\
\hspace{-0.3cm}&=&\hspace{-0.3cm}
q\int_{x}^{b}\frac{W_{q}(x-c)}{ W_{q}(t-c)}\left(W(0)h(t)+\int_{0}^{\infty}n\left(\mathrm{e}^{-qs}h(t-\varepsilon(s))\mathbf{1}_{\{s<\zeta\wedge\rho^{+}_{t-c}\}}\right)\mathrm{d}s\right)\mathrm{d}t,
\end{eqnarray}
where the identity
$$\mathbb{E}_{x-c}\left(\mathrm{e}^{-q \tau_{t-c}^{+}};\tau_{t-c}^{+}<\tau_{0}^{-}\right)=\frac{W_{q}(x-c)}{ W_{q}(t-c)},\quad -\infty<c\leq x\leq t<\infty,$$
is used.
Equating the right hand sides of (\ref{h1}) and (\ref{h2}) and then differentiating the resulting equation with respect to $b$ gives
\begin{eqnarray}
\hspace{-0.3cm}&&\hspace{-0.3cm}
\frac{W_{q}(x-c)}{ W_{q}(b-c)}\left(W(0)h(b)+\int_{0}^{\infty}n\left(\mathrm{e}^{-qs}h(b-\varepsilon(s))\mathbf{1}_{\{s<\zeta\wedge\rho^{+}_{b-c}\}}\right)\mathrm{d}s\right)
\nonumber\\
\hspace{-0.3cm}&=&\hspace{-0.3cm}
\frac{W_{q}(x-c)}{ W_{q}(b-c)}\left(h(b)W(0)+\int_{c}^{b}h(y)\left(W_{q}^{\prime}(b-y)-\frac{W_{q}^{\prime}(b-c)}{W_{q}(b-c)}W_{q}(b-y)\right)\mathrm{d}y\right),\nonumber
\end{eqnarray}
or equivalently,
\begin{eqnarray}\label{impo.iden.for.n.01}
\hspace{-0.3cm}&&\hspace{-0.3cm}
\int_{0}^{\infty}n\left(\mathrm{e}^{-qs}h(b-\varepsilon(s))\mathbf{1}_{\{s<\zeta\wedge\rho^{+}_{b-c}\}}\right)\mathrm{d}s
\nonumber\\
\hspace{-0.3cm}&=&\hspace{-0.3cm}
\int_{c}^{b}h(y)\left(W_{q}^{\prime}(b-y)-\frac{W_{q}^{\prime}(b-c)}{W_{q}(b-c)}W_{q}(b-y)\right)\mathrm{d}y.
\end{eqnarray}
Combining (\ref{impo.iden.for.n.01}) and (\ref{qg1.raw.}), we get
\begin{eqnarray}\label{qg1}
g_{1}(x)\hspace{-0.3cm}&=&\hspace{-0.3cm}\int_{x}^{b}\mathrm{e}^{-\int_{x}^{s}\frac{W_{q}^{\prime}(\overline{\xi}(z))}{ W_{q}(\overline{\xi}(z))}\mathrm{d}z}
\int_{\xi(s)}^{s}h(y)\left(W_{q}^{\prime}(s-y)-\frac{W_{q}^{\prime}(\overline{\xi}(s))}{W_{q}(\overline{\xi}(s))}W_{q}(s-y)\right)\mathrm{d}y
\mathrm{d}s.
\end{eqnarray}

Using the memoryless property of exponential random variable, $qg_{2}(x)$ can be rewritten as
\begin{eqnarray}\label{qg2}
qg_{2}(x)\hspace{-0.3cm}&=&\hspace{-0.3cm}\mathbb{E}_{x}\left(h(U(e_{q}))\mathbf{1}_{\{U(e_{q})<\overline{U}(e_{q}),\,\tau_{\xi}<e_{q}<T_{1}<\kappa_{b}^{+}\wedge \theta_{\lambda}\}}\right)
\nonumber\\
\hspace{-0.3cm}&&\hspace{-0.3cm}
+\mathbb{E}_{x}\left(h(U(e_{q}))\mathbf{1}_{\{U(e_{q})<\overline{U}(e_{q}),\,\tau_{\xi}<T_{1}<e_{q}<\kappa_{b}^{+}\wedge \theta_{\lambda}\}}\right)
\nonumber\\
\hspace{-0.3cm}&=&\hspace{-0.3cm}\mathbb{E}_{x}\left(h(U(e_{q}))\mathbf{1}_{\{\tau_{\xi}<e_{q}<T_{1}<\kappa_{b}^{+}\wedge \theta_{\lambda}\}}\right)
+\mathbb{E}_{x}\left(h(U(e_{q}))\mathbf{1}_{\{U(e_{q})<\overline{U}(e_{q}),\,\tau_{\xi}<T_{1}<e_{q}<\kappa_{b}^{+}\wedge \theta_{\lambda}\}}\right)
\nonumber\\
\hspace{-0.3cm}&=:&\hspace{-0.3cm}\,qg_{21}(x)+qg_{22}(x),
\end{eqnarray}
where we also took use of the fact that $\tau_{\xi}<\kappa_{b}^{+}$ implies $T_{1}<\kappa_{b}^{+}$ as in \eqref{part2}.

Using once again the facts that $\tau_{\xi}<T_{1}$ holds implicitly and $\tau_{\xi}<\kappa_{b}^{+}$ implies $T_{1}<\kappa_{b}^{+}$,
we have by (\ref{reso.meas.Y}) and \eqref{12}
\begin{eqnarray}\label{qg21}
qg_{21}(x)
\hspace{-0.3cm}&=&\hspace{-0.3cm}\mathbb{E}_{x}\left(\mathbb{E}_{x}\left(\left.h(U(e_{q}))\mathbf{1}_{\{\tau_{\xi}<e_{q}<T_{1}<\kappa_{b}^{+}\wedge \theta_{\lambda}\}}\right|\mathcal{F}_{\tau_{\xi}}\right)\right)
\nonumber\\
\hspace{-0.3cm}&=&\hspace{-0.3cm}
\mathbb{E}_{x}\left(\mathbf{1}_{\{\tau_{\xi}<e_{q}\wedge\kappa_{b}^{+}\}}\mathbb{E}_{x}\left(\left.h(U(e_{q}))\mathbf{1}_{\{e_{q}<T_{1}< \theta_{\lambda}\}}\right|\mathcal{F}_{\tau_{\xi}}\right)\right)
\nonumber\\
\hspace{-0.3cm}&=&\hspace{-0.3cm}
\mathbb{E}_{x}\left(\mathbf{1}_{\{\tau_{\xi}<e_{q}\wedge\kappa_{b}^{+}\}}\left.\mathbb{E}_{}\left(h(\xi(z)+Y(e_{q}))\mathbf{1}_{\{e_{q}<\sigma^{+}_{\overline{\xi}(z)}<e_{\lambda}\}}\right)\right|_{z=\overline{X}(\tau_{\xi})}\right)
\nonumber\\
\hspace{-0.3cm}&=&\hspace{-0.3cm}
q\mathbb{E}_{x}\left(\mathrm{e}^{-q \tau_{\xi}}\mathbf{1}_{\{\tau_{\xi}<\tau_{b}^{+}\}}\int_{0}^{\overline{\xi}(\overline{X}(\tau_{\xi}))}\frac{Z_{\lambda}(y)\,h(\xi(\overline{X}(\tau_{\xi}))+y)}{Z_{\lambda}(\overline{\xi}(\overline{X}(\tau_{\xi})))}
\frac{W_{q+\lambda}(\overline{\xi}(\overline{X}(\tau_{\xi}))-y)}{Z_{q+\lambda}(\overline{\xi}(\overline{X}(\tau_{\xi})))}\mathrm{d}y\right)
\nonumber\\
\hspace{-0.3cm}&=&\hspace{-0.3cm}
q\int_{x}^{b}
\exp\left(-\int_{x}^{s}\frac{W_{q}^{\prime}(\overline{\xi}(z))}{ W_{q}(\overline{\xi}(z))}\mathrm{d}z\right)
\left(\frac{W_{q}^{\prime}(\overline{\xi}(s))}{W_{q}(\overline{\xi}(s))}Z_{q}(\overline{\xi}(s))-qW_{q}(\overline{\xi}(s))\right)
\nonumber\\
\hspace{-0.3cm}&&\hspace{-0.3cm}
\,\,\,\,\,\,\,\,\,\times\int_{0}^{\overline{\xi}(s)}\frac{Z_{\lambda}(y)h(\xi(s)+y)
W_{q+\lambda}(\overline{\xi}(s)-y)}
{Z_{\lambda}(\overline{\xi}(s))Z_{q+\lambda}(\overline{\xi}(s))}\mathrm{d}y\mathrm{d}s,
\end{eqnarray}
where for  the fourth equation the following equation is applied
\begin{eqnarray}
\hspace{-0.3cm}&&\hspace{-0.3cm}
\mathbb{E}_{}\left(h(Y(e_{q}))\mathbf{1}_{\{e_{q}<\sigma^{+}_{a}<e_{\lambda}\}}\right)
\nonumber\\
\hspace{-0.3cm}&=&\hspace{-0.3cm}
\mathbb{E}_{}\left(\mathrm{e}^{-\lambda \sigma^{+}_{a}}h(Y(e_{q}))\mathbf{1}_{\{e_{q}<\sigma^{+}_{a}\}}\right)
\nonumber\\
\hspace{-0.3cm}&=&\hspace{-0.3cm}
q\int_{0}^{\infty}\mathrm{e}^{-(q+\lambda)t}\,\mathbb{E}_{}\left(h(Y(t))\mathbf{1}_{\{t<\sigma^{+}_{a}\}}\mathbb{E}_{}\left(\left.\mathrm{e}^{-\lambda \left(\sigma^{+}_{a}-t\right)}\right|\mathcal{F}_{t}\right)\right)
\mathrm{d}t
\nonumber\\
\hspace{-0.3cm}&=&\hspace{-0.3cm}
q\int_{0}^{\infty}\mathrm{e}^{-(q+\lambda)t}\,\mathbb{E}_{}\left(\frac{Z_{\lambda}(Y(t))}{Z_{\lambda}(a)}\,h(Y(t))\mathbf{1}_{\{t<\sigma^{+}_{a}\}}\right)
\mathrm{d}t, \quad a\in(0,\infty).\nonumber
\end{eqnarray}

In addition, observing that $\tau_\xi<T_1$ for $\tau_\xi<\infty$, by (\ref{two.sid.exit.Y}) and \eqref{12} one can rewrite $qg_{22}(x)$ as
\begin{eqnarray}\label{qg22}
qg_{22}(x)\hspace{-0.3cm}&=&\hspace{-0.3cm}
\mathbb{E}_{x}\left(\mathbb{E}_{x}\left(\left.h(U(e_{q}))\mathbf{1}_{\{U(e_{q})<\overline{U}(e_{q}),\,\tau_{\xi}<T_{1}<e_{q}<\kappa_{b}^{+}\wedge \theta_{\lambda}\}}\right|\mathcal{F}_{T_{1}}\right)\right)
\nonumber\\
\hspace{-0.3cm}&=&\hspace{-0.3cm}
\mathbb{E}_{x}\left(\mathbf{1}_{\{T_{1}<e_{q}\wedge\kappa_{b}^{+},\,T_{1}-\tau_{\xi}<e_{\lambda}\}}
\left(qg_{1}(\overline{X}(\tau_{\xi}))+qg_{2}(\overline{X}(\tau_{\xi}))\right)\right)
\nonumber\\
\hspace{-0.3cm}&=&\hspace{-0.3cm}
\mathbb{E}_{x}\left(\mathbb{E}_{x}\left(\left.\mathrm{e}^{-qT_{1}}
\mathrm{e}^{-\lambda (T_{1}-\tau_{\xi})}\mathbf{1}_{\{T_{1}<\kappa_{b}^{+}\}}
\left(qg_{1}(\overline{X}(\tau_{\xi}))+qg_{2}(\overline{X}(\tau_{\xi}))\right)\right|\mathcal{F}_{\tau_{\xi}}\right)\right)
\nonumber\\
\hspace{-0.3cm}&=&\hspace{-0.3cm}
\mathbb{E}_{x}\left(\mathrm{e}^{-q\tau_{\xi}}\mathbf{1}_{\{\tau_{\xi}<\tau_{b}^{+}\}}\frac{1}{Z_{q+\lambda}(\overline{\xi}(\overline{X}(\tau_{\xi})))}\left(qg_{1}(\overline{X}(\tau_{\xi}))+qg_{2}(\overline{X}(\tau_{\xi}))\right)\right)
\nonumber\\
\hspace{-0.3cm}&=&\hspace{-0.3cm}
q\int_{x}^{b}
\frac{g_{1}(s)+g_{2}(s)}{Z_{q+\lambda}(\overline{\xi}(s))}\mathrm{e}^{-\int_{x}^{s}\frac{W_{q}^{\prime}(\overline{\xi}(z))}{ W_{q}(\overline{\xi}(z))}\mathrm{d}z}
\left(\frac{W_{q}^{\prime}(\overline{\xi}(s))}{W_{q}(\overline{\xi}(s))}Z_{q}(\overline{\xi}(s))-qW_{q}(\overline{\xi}(s))\right)\mathrm{d}s.
\end{eqnarray}

Combining (\ref{qg3}), \eqref{25}, (\ref{qg1}), (\ref{qg2}), (\ref{qg21}) and (\ref{qg22}), we obtain the following differential equation on $g(x)$ with boundary condition $g(b)=0$.
\begin{eqnarray}\label{g'}
g^{\prime}(x)\hspace{-0.3cm}&=&\hspace{-0.3cm}\frac{W_{q}^{\prime}(\overline{\xi}(x))}{ W_{q}(\overline{\xi}(x))}g(x)- W(0)h(x)
-\frac{g_{3}(x)+g_{4}(x)}{Z_{q+\lambda}(\overline{\xi}(x))}
\left(\frac{W_{q}^{\prime}(\overline{\xi}(x))}{W_{q}(\overline{\xi}(x))}
Z_{q}(\overline{\xi}(x))-qW_{q}(\overline{\xi}(x))\right)
\nonumber\\
\hspace{-0.3cm}&&\hspace{-0.3cm}
-\int_{\xi(x)}^{x}h(y)\left(W_{q}^{\prime}(x-y)-\frac{W_{q}^{\prime}(\overline{\xi}(x))}{W_{q}(\overline{\xi}(x))}W_{q}(x-y)\right)\mathrm{d}y
\nonumber\\
\hspace{-0.3cm}&&\hspace{-0.3cm}
-\frac{g_{1}(x)+g_{2}(x)}{Z_{q+\lambda}(\overline{\xi}(x))}
\left(\frac{W_{q}^{\prime}(\overline{\xi}(x))}{W_{q}(\overline{\xi}(x))}Z_{q}(\overline{\xi}(x))-qW_{q}(\overline{\xi}(x))\right)
\nonumber\\
\hspace{-0.3cm}&&\hspace{-0.3cm}
-\left(\frac{W_{q}^{\prime}(\overline{\xi}(x))}{W_{q}(\overline{\xi}(x))}Z_{q}(\overline{\xi}(x))-qW_{q}(\overline{\xi}(x))\right)
\int_{0}^{\overline{\xi}(x)}\frac{Z_{\lambda}(y)h(\xi(x)+y)
W_{q+\lambda}(\overline{\xi}(x)-y)}
{Z_{\lambda}(\overline{\xi}(x))Z_{q+\lambda}(\overline{\xi}(x))}\mathrm{d}y
\nonumber\\
\hspace{-0.3cm}&=&\hspace{-0.3cm}
\ell_{1}(x)g(x)- W(0)h(x)
-\int_{\xi(x)}^{x}h(y)\left(W_{q}^{\prime}(x-y)-\frac{W_{q}^{\prime}(\overline{\xi}(x))}{W_{q}(\overline{\xi}(x))}W_{q}(x-y)\right)\mathrm{d}y
\nonumber\\
\hspace{-0.3cm}&&\hspace{-0.3cm}
-\left(\frac{W_{q}^{\prime}(\overline{\xi}(x))}{W_{q}(\overline{\xi}(x))}Z_{q}(\overline{\xi}(x))-qW_{q}(\overline{\xi}(x))\right)
\int_{\xi(x)}^{x}\frac{Z_{\lambda}(y-\xi(x))h(y)
W_{q}(x-y)}
{Z_{\lambda}(\overline{\xi}(x))Z_{q+\lambda}(\overline{\xi}(x))}
\mathrm{d}y.
\end{eqnarray}
 Solving equation (\ref{g'}) yields
\begin{eqnarray}
\label{eq.gene.reso.meas.}
g(x)\hspace{-0.3cm}&=&\hspace{-0.3cm}\int_{0}^{\infty}\mathrm{e}^{-qt}\mathbb{E}_{x}\left(h(U(t));t<\kappa_{b}^{+}\wedge \theta_{\lambda}\right)\mathrm{d}t
\nonumber\\
\hspace{-0.3cm}&=&\hspace{-0.3cm}
W_{q}(0)\int_{x}^{b}h(y)
\exp\left(-\int_{x}^{y}
\ell_{1}(w)
\mathrm{d}w\right)\mathrm{d}y
+\int_{x}^{b}\exp\left(-\int_{x}^{y}
\ell_{1}(w)
\mathrm{d}w\right)
\nonumber\\
\hspace{-0.3cm}&&\hspace{-0.3cm}
\times\left(\int_{\xi(y)}^{y}h(z)\left(W_{q}^{\prime}(y-z)-\frac{W_{q}^{\prime}(\overline{\xi}(y))}{W_{q}(\overline{\xi}(y))}W_{q}(y-z)\right)\mathrm{d}z\right.
\nonumber\\
\hspace{-0.3cm}&&\hspace{-0.3cm}
\left.+
\left(\frac{W_{q}^{\prime}(\overline{\xi}(y))}{W_{q}(\overline{\xi}(y))}Z_{q}(\overline{\xi}(y))-qW_{q}(\overline{\xi}(y))\right)
\int_{\xi(y)}^{y}\frac{Z_{\lambda}(z-\xi(y))h(z)
W_{q}(y-z)}
{Z_{\lambda}(\overline{\xi}(y))Z_{q+\lambda}(\overline{\xi}(y))}
\mathrm{d}z\right)  \mathrm{d}y.
\end{eqnarray}
The expression (\ref{resovent.meas.}) follows immediately from (\ref{eq.gene.reso.meas.}).
\end{proof}

\medskip
The following result gives an expression of the expectation of the total discounted capital injections until time $\kappa_{b}^{+}\wedge \theta_{\lambda}$.

\medskip
\begin{thm}\label{3.3}
For $q\in (0,\infty)$, we have
\begin{eqnarray}\label{expr.of.exp.tot.dis.cap.inj.}
V_\xi(x;b)
\hspace{-0.3cm}&=&\hspace{-0.3cm}
\int_{x}^{b}
\exp\left(-\int_{x}^{y}
\ell_{1}(w)
\mathrm{d}w\right)
\ell_{4}(y)\mathrm{d}y,\quad x\in(-\infty,b],
\end{eqnarray}
where $\ell_{1}$ is defined as in Theorem \ref{3.1}, and
\begin{eqnarray}
\ell_{4}(y)\hspace{-0.3cm}&=&\hspace{-0.3cm}Z_{q}(\overline{\xi}(y))-\psi^{\prime}(0+)W_{q}(\overline{\xi}(y))-\frac{\overline{Z}_{q}(\overline{\xi}(y))-\psi^{\prime}(0+)\overline{W}_{q}(\overline{\xi}(y))}{W_{q}(\overline{\xi}(y))}W_{q}^{\prime}(\overline{\xi}(y))
\nonumber\\
\hspace{-0.3cm}&&\hspace{-0.3cm}
+
\frac{\overline{Z}_{q+\lambda}(\overline{\xi}(y))-\psi^{\prime}(0+)\overline{W}_{q+\lambda}(\overline{\xi}(y))}{Z_{q+\lambda}(\overline{\xi}(y))}
\left(\frac{W_{q}^{\prime}(\overline{\xi}(y))}{W_{q}(\overline{\xi}(y))}Z_{q}(\overline{\xi}(y))-qW_{q}(\overline{\xi}(y))\right).\nonumber
\end{eqnarray}
\end{thm}

\begin{proof}[Proof:]\,\,\,
For $q\in(0,\infty)$ and $x\in(-\infty,b]$, we have
\begin{eqnarray}\label{disc.total.cap.}
V_\xi(x;b)\hspace{-0.3cm}&=&\hspace{-0.3cm}\mathbb{E}_{x}\left(\int_{0}^{\kappa_{b}^{+}\wedge \theta_{\lambda}}\mathrm{e}^{-q t}\mathrm{d}R(t)\right)
=\mathbb{E}_{x}\left(\mathrm{e}^{-q\tau_{\xi}}\mathbf{1}_{\{\tau_{\xi}<\tau_{b}^{+}\}}\left(\xi(\overline{X}(\tau_{\xi}))-X(\tau_{\xi})\right)\right)
\nonumber\\
\hspace{-0.3cm}&&\hspace{-0.3cm}
+\mathbb{E}_{x}\left(\mathbf{1}_{\{\tau_{\xi}<\tau_{b}^{+}\}}\int_{\tau_{\xi}+}^{T_{1}\wedge \theta_{\lambda}}\mathrm{e}^{-q t}\mathrm{d}R(t)\right)
+\mathbb{E}_{x}\left(\mathbf{1}_{\{\tau_{\xi}<\kappa_{b}^{+}\}}\int_{T_{1}\wedge \theta_{\lambda}}^{\kappa_{b}^{+}\wedge \theta_{\lambda}}\mathrm{e}^{-q t}\mathrm{d}R(t)\right)
\nonumber\\
\hspace{-0.3cm}&=:&\hspace{-0.3cm}V_{1}(x;b)+V_{2}(x;b)+V_{3}(x;b).
\end{eqnarray}

By \eqref{14}, $V_{1}(x;b)$ can be expressed as
\begin{eqnarray}\label{l1}
V_{1}(x;b)
\hspace{-0.3cm}&=&\hspace{-0.3cm}\int_{x}^{b}
\exp\left(-\int_{x}^{s}\frac{W_{q}^{\prime}(\overline{\xi}\left(z\right))}
{W_{q}(\overline{\xi}\left(z\right))}\mathrm{d}z\right)
\Bigg(Z_{q}(\overline{\xi}(s))-\psi^{\prime}(0+)W_{q}(\overline{\xi}(s))
\nonumber\\
\hspace{-0.3cm}&&\hspace{-0.3cm}\hspace{4cm}
-\frac{\overline{Z}_{q}(\overline{\xi}(s))-\psi^{\prime}(0+)\overline{W}_{q}(\overline{\xi}(s))}{W_{q}(\overline{\xi}(s))}W_{q}^{\prime}(\overline{\xi}(s))\Bigg)\mathrm{d}s.
\end{eqnarray}
By the Markov property for the reflected process $(U,\overline{U})$,
\begin{eqnarray}
\label{v0(0}
V_{2}(x;b)\hspace{-0.3cm}&=&\hspace{-0.3cm}
\mathbb{E}_{x}\left(\mathbf{1}_{\{\tau_{\xi}<\tau_{b}^{+},\,e_{\lambda}^{(1)}\geq T_{1}-\tau_{\xi}\}}\int_{\tau_{\xi}+}^{T_{1}}\mathrm{e}^{-q t}\mathrm{d}R(t)\right)
+\mathbb{E}_{x}\left(\mathbf{1}_{\{\tau_{\xi}<\tau_{b}^{+},\,e_{\lambda}^{(1)}<T_{1}-\tau_{\xi}\}}\int_{\tau_{\xi}+}^{\tau_{\xi}+e_{\lambda}^{(1)}}\mathrm{e}^{-q t}\mathrm{d}R(t)\right)
\nonumber\\
\hspace{-0.3cm}&=&\hspace{-0.3cm}
\mathbb{E}_{x}\left(\mathrm{e}^{-q\tau_{\xi}}\mathbf{1}_{\{\tau_{\xi}<\tau_{b}^{+}\}}\left(\mathbb{E}\left(\left.\mathrm{e}^{-\lambda \sigma_{z}^{+}}\int_{0}^{\sigma_{z}^{+}}
\mathrm{e}^{-qt}\mathrm{d}\left(-\underline{X}(t)\right)\right)\right|_{z=\overline{\xi}(\overline{X}(\tau_{\xi}))}
\right)\right)
\nonumber\\
\hspace{-0.3cm}&&\hspace{-0.3cm}
+\mathbb{E}_{x}\left(\mathrm{e}^{-q\tau_{\xi}}\mathbf{1}_{\{\tau_{\xi}<\tau_{b}^{+}\}}\left(\mathbb{E}\left(\left.\mathbf{1}_{\{e_{\lambda}<\sigma_{z}^{+}\}}\int_{0}^{e_{\lambda}}
\mathrm{e}^{-qt}\mathrm{d}\left(-\underline{X}(t)\right)\right)\right|_{z=\overline{\xi}(\overline{X}(\tau_{\xi}))}
\right)\right)
\nonumber\\
\hspace{-0.3cm}&=&\hspace{-0.3cm}
\mathbb{E}_{x}\left(\mathrm{e}^{-q\tau_{\xi}}\mathbf{1}_{\{\tau_{\xi}<\tau_{b}^{+}\}}\left(-\frac{\psi^{\prime}(0+)}{q+\lambda}+\frac{\overline{Z}_{q+\lambda}(\overline{\xi}(\overline{X}(\tau_{\xi})))+\frac{\psi^{\prime}(0+)}{q+\lambda}}{Z_{q+\lambda}(\overline{\xi}(\overline{X}(\tau_{\xi})))}
\right)\right),
\end{eqnarray}
where the following equations are applied
\begin{eqnarray}
\hspace{-0.3cm}&&\hspace{-0.3cm}
\mathbb{E}\left(\mathbf{1}_{\{e_{\lambda}<\sigma_{z}^{+}\}}\int_{0}^{e_{\lambda}}
\mathrm{e}^{-qt}\mathrm{d}\left(-\underline{X}(t)\right)\right)
\nonumber\\
\hspace{-0.3cm}&=&\hspace{-0.3cm}
\mathbb{E}\left(\int_{0}^{\sigma_{z}^{+}} \int_{0}^{t}\mathrm{e}^{-q s}\mathrm{d}\left(-\underline{X}(s)\right)\mathrm{d}\left(-\mathrm{e}^{-\lambda t}\right)\right)
\nonumber\\
\hspace{-0.3cm}&=&\hspace{-0.3cm}
- \mathbb{E}\left(\mathrm{e}^{-\lambda \sigma_{z}^{+}}
\int_{0}^{\sigma_{z}^{+}}\mathrm{e}^{-q t}\mathrm{d}\left(-\underline{X}(t)\right)\right)
+
\mathbb{E}\left(\int_{0}^{\sigma_{z}^{+}}
\mathrm{e}^{-(q+\lambda) t}\mathrm{d}\left(-\underline{X}(t)\right)\right),\quad z\in(0,\infty),\nonumber
\end{eqnarray}
and
$$\mathbb{E}\left(\int_{0}^{\sigma_{z}^{+}}
\mathrm{e}^{-(q+\lambda)t}\mathrm{d}\left(-\underline{X}(t)\right)\right)=-\frac{\psi^{\prime}(0+)}{q+\lambda}+\frac{\overline{Z}_{q+\lambda}(z)+\frac{\psi^{\prime}(0+)}{q+\lambda}}{Z_{q+\lambda}(z)},\quad z\in(0,\infty),$$
which can be found in the proof of Theorem 1 of Avram et al. (2007).
Combining \eqref{v0(0} and \eqref{12} yields
\begin{eqnarray}
\label{l2.}
V_{2}(x;b)
\hspace{-0.3cm}&=&\hspace{-0.3cm}
\int_{x}^{b}\left(-\frac{\psi^{\prime}(0+)}{q+\lambda}+\frac{\overline{Z}_{q+\lambda}(\overline{\xi}(s))+\frac{\psi^{\prime}(0+)}{q+\lambda}}{Z_{q+\lambda}(\overline{\xi}(s))}\right)
\exp\left(-\int_{x}^{s}\frac{W_{q}^{\prime}(\overline{\xi}(z))}{ W_{q}(\overline{\xi}(z))}\mathrm{d}z\right)
\nonumber\\
\hspace{-0.3cm}&&\hspace{-0.3cm}\,\,\,\,\,\,\,\times
\left(\frac{W_{q}^{\prime}(\overline{\xi}(s))}{W_{q}(\overline{\xi}(s))}Z_{q}(\overline{\xi}(s))-qW_{q}(\overline{\xi}(s))\right)\mathrm{d}s.
\end{eqnarray}

Making use of \eqref{12}  again, one can get
\begin{eqnarray}\label{l3}
V_{3}(x;b)\hspace{-0.3cm}&=&\hspace{-0.3cm}
\mathbb{E}_{x}\left(\mathbb{E}_{x}\left(\left.\mathbf{1}_{\{\tau_{\xi}<\kappa_{b}^{+},\,T_{1}-\tau_{\xi}<e_{\lambda}^{(1)}\}}\int_{T_{1}}^{\kappa_{b}^{+}\wedge \theta_{\lambda}}\mathrm{e}^{-q t}\mathrm{d}R(t)\right|\mathcal{F}_{T_{1}}\right)\right)
\nonumber\\
\hspace{-0.3cm}&=&\hspace{-0.3cm}
\mathbb{E}_{x}\left(\mathbb{E}_{x}\left(\left.\mathrm{e}^{-q\tau_{\xi}}\,\mathrm{e}^{-(q+\lambda)(T_{1}-\tau_{\xi})}\mathbf{1}_{\{\tau_{\xi}<\kappa_{b}^{+}\}}V_\xi(\overline{X}(\tau_{\xi});b)\right|\mathcal{F}_{\tau_{\xi}}\right)\right)
\nonumber\\
\hspace{-0.3cm}&=&\hspace{-0.3cm}
\mathbb{E}_{x}\left(\mathrm{e}^{-q\tau_{\xi}}\mathbf{1}_{\{\tau_{\xi}<\tau_{b}^{+}\}}\frac{V_\xi(\overline{X}(\tau_{\xi});b)}{Z_{q+\lambda}(\overline{\xi}(\overline{X}(\tau_{\xi})))}\right)
\nonumber\\
\hspace{-0.3cm}&=&\hspace{-0.3cm}
\int_{x}^{b}\frac{V_\xi(s;b)}{Z_{q+\lambda}(\overline{\xi}(s))}
\mathrm{e}^{-\int_{x}^{s}\frac{W_{q}^{\prime}(\overline{\xi}(z))}{ W_{q}(\overline{\xi}(z))}\mathrm{d}z}
\left(\frac{W_{q}^{\prime}(\overline{\xi}(s))}{W_{q}(\overline{\xi}(s))}Z_{q}(\overline{\xi}(s))-qW_{q}(\overline{\xi}(s))\right)\mathrm{d}s.
\end{eqnarray}

Denote by $V_{\xi}^{\prime}(x;b)$ the derivative of $V_{\xi}(x;b)$ with respect to its first argument. Combining (\ref{disc.total.cap.}), (\ref{l1}), (\ref{l2.}) and (\ref{l3}) we have
\begin{eqnarray}\label{l'}
V_{\xi}^{\prime}(x;b)
\hspace{-0.3cm}&=&\hspace{-0.3cm}
\frac{W_{q}^{\prime}(\overline{\xi}(x))}{ W_{q}(\overline{\xi}(x))}V_\xi(x;b)
-\frac{V_\xi(x;b)}{Z_{q+\lambda}(\overline{\xi}(x))}
\left(\frac{W_{q}^{\prime}(\overline{\xi}(x))}{W_{q}(\overline{\xi}(x))}Z_{q}(\overline{\xi}(x))-qW_{q}(\overline{\xi}(x))\right)
\nonumber\\
\hspace{-0.3cm}&&\hspace{-0.3cm}
-\left(Z_{q}(\overline{\xi}(x))-\psi^{\prime}(0+)W_{q}(\overline{\xi}(x))-\frac{\overline{Z}_{q}(\overline{\xi}(x))-\psi^{\prime}(0+)\overline{W}_{q}(\overline{\xi}(x))}{W_{q}(\overline{\xi}(x))}W_{q}^{\prime}(\overline{\xi}(x))\right)
\nonumber\\
\hspace{-0.3cm}&&\hspace{-0.3cm}
-\left(-\frac{\psi^{\prime}(0+)}{q+\lambda}+\frac{\overline{Z}_{q+\lambda}(\overline{\xi}(x))+\frac{\psi^{\prime}(0+)}{q+\lambda}}{Z_{q+\lambda}(\overline{\xi}(x))}\right)\left(\frac{W_{q}^{\prime}(\overline{\xi}(x))}{W_{q}(\overline{\xi}(x))}Z_{q}(\overline{\xi}(x))-qW_{q}(\overline{\xi}(x))\right)
\nonumber\\
\hspace{-0.3cm}&=&\hspace{-0.3cm}
\ell_{1}(x)V_\xi(x;b)-\ell_{4}(x).
\end{eqnarray}
Solving (\ref{l'}) with boundary condition $V_\xi(b, b)=0 $,  we obtain (\ref{expr.of.exp.tot.dis.cap.inj.}).
\end{proof}

\section{Numerical examples}\label{4}
To exemplify the main results, we discuss some numerical examples for one-sided jump-diffusion process $X$ with Laplace exponent $\psi(\theta)=\mu\theta+\frac{\sigma^2}{2}\theta^2-\frac{a\theta}{\theta+c}$ for all $\theta\in\mathbb{R}$ s.t. $\theta\neq -c$. See Example \ref{ex:Ex1} for the corresponding scale function. We set $\mu=0.075$, $a=0.5$ and $c=9$ (on average once every two years the firm suffers an instantaneous loss of $10\%$ of its value), and $q=5\%$. 

Figures \ref{fig:exit} and \ref{fig:Vxi} display various shapes of ruin probability $\mathbb{P}_x(\theta_{\lambda}<\infty)$ (\ref{eq:ruinprobab}) and the expected nett present value $V_{\xi}(x):=V_{\xi}(x;b=\infty)$ of the total amount of required capital injection to the firm as a function of initial surplus $x$ and the default monitoring frequency $\lambda$. The computation was performed for different forms of drawdown function: $\xi(x)=Kx$, with $0<K<1$, and $\xi(x)=\min\{1,Kx\}$. The former dictates that default is announced and followed by capital injection as soon as the surplus process has crossed below $K\%$ of its last record high, whereas the latter deals with the case of injecting capital to the firm as soon as the surplus is less than one dollar or below $K\%$ of the last record high, whichever is smaller. The results are presented for two cases: $\sigma=0.2$ (the process has paths of unbounded variation) and $\sigma=0$ (paths with bounded variation).

Over all, we observe that the ruin probability $\mathbb{P}_x(\theta_{\lambda}<\infty)$ decreases as the initial surplus increases. This implies that firms with higher initial endowment/surplus has lower probability of ruin than those with lower value of surplus at the beginning. Furthermore, when $\sigma\neq 0$ in which case the firm has additional (immediate) exposure to risk, say from investing in financial market, the ruin probability is higher than those firms which do not have any risk exposure ($\sigma=0$) other than the claim from the insurance holder, which arrives at exponential random time. From the sample paths point of view, the presence of Brownian motion allows the paths to reach a new maximum level and then makes an excursion below that level before a jump arrives. As a result, such movement triggers the observation clock start to run and put the firm into a risky position of getting default. Moreover, we also notice from the figure that the higher the observation frequency $\lambda$, the higher the ruin probability. This is to say that the larger $\lambda$ gives the firm lesser time to come out of dilution period during which the firm is in financial distress, resulting in a higher chance of going ruin. The choice of drawdown level $\xi(x)$ also determines the shape of the ruin probability. Under the drawdown level $Kx$, the ruin probability is higher for larger value of $K$. This is due to the fact that the higher value of $K$ sets the default level higher causing the firm to go default/ruin sooner than lower value of $K$, in particular when $\sigma\neq 0$. As $\min\{1,Kx\}\leq Kx$, default is expected to occur earlier under drawdown level $\xi(x)=Kx$ than  $\xi(x)=\min\{1,Kx\}$ for the same reason explained. For the latter, the ruin probability decreases linearly when the surplus is about less than one unit, then decreases at exponential rate when the surplus is larger than that value. Similar observation is observed for $\sigma=0$ with only exception that the curve has lower degree of smoothness than the case $\sigma\neq 0$ in which case the scale function is twice continuously differentiable over $(0,\infty)$. 

Similar feature is exhibited by the expected total amount of discounted cash for capital injection function $V_{\xi}(x)$. Healthier firm with larger initial surplus $x$ requires less capital injection in total than that of unhealthier firm with lower surplus. When the firm has more uncertainties as a result of investing in financial market, the case $\sigma\neq 0$, the firm requires more capital injection than those firm which do not have riskier exposure ($\sigma=0$) other than the claim from the insurance holder arriving in exponential random time. As explained above, the presence of Brownian motion puts the firm in riskier situation with higher probability of ruin when the surplus process makes an excursion from its last record high. This induces the firm to ask for more financial coverage (injection) to deal with during the dilution period. Moreover, the longer the observation window (the lower the value of $\lambda$), the firm receives more capital injection from stakeholder than otherwise. The choice of drawdown level $\xi(x)$ also shifts the curve $\xi(x)$. For $\xi(x)=Kx$, the curve is higher for larger value of $K$. The upward shift of the curve is attributed by the fact that the larger value of $\lambda$ increases the default threshold to higher level giving the firm higher chance of default resulting the firm in asking for more capital to prevent ruin to occur at earlier stage, in particular when $\sigma\neq 0$. As $\min\{1,Kx\}\leq Kx$, default is expected to occur later for $\xi(x)=\min\{1,Kx\}$ than that of under the drawdown level $\xi(x)=Kx$ which in turn requiring lesser amount of capital injection. Under $\xi(x)=\min\{1,Kx\}$, the function $V_{\xi}(x)$ decreases linearly with quite high slope (rate) when the surplus is about less than one unit, then decreases exponentially at lower rate when the surplus is larger than that value. Overall, unlike the ruin probability $\mathbb{P}_x(\theta_{\lambda}<\infty)$, the function $V_{\xi}(x)$ retains its convexity in all cases having zero value at infinity implying the function to be positive for all $x$.

The above analyses conclude our numerical study on the ruin probability and nett present value of the capital injection which summarizes their various shape w.r.t changing the value of observation frequency $\lambda$ and initial surplus $x$.

\begin{figure}[t!]
\centering
\begin{tabular}{cc}
\subf{\includegraphics[width=0.465\textwidth]{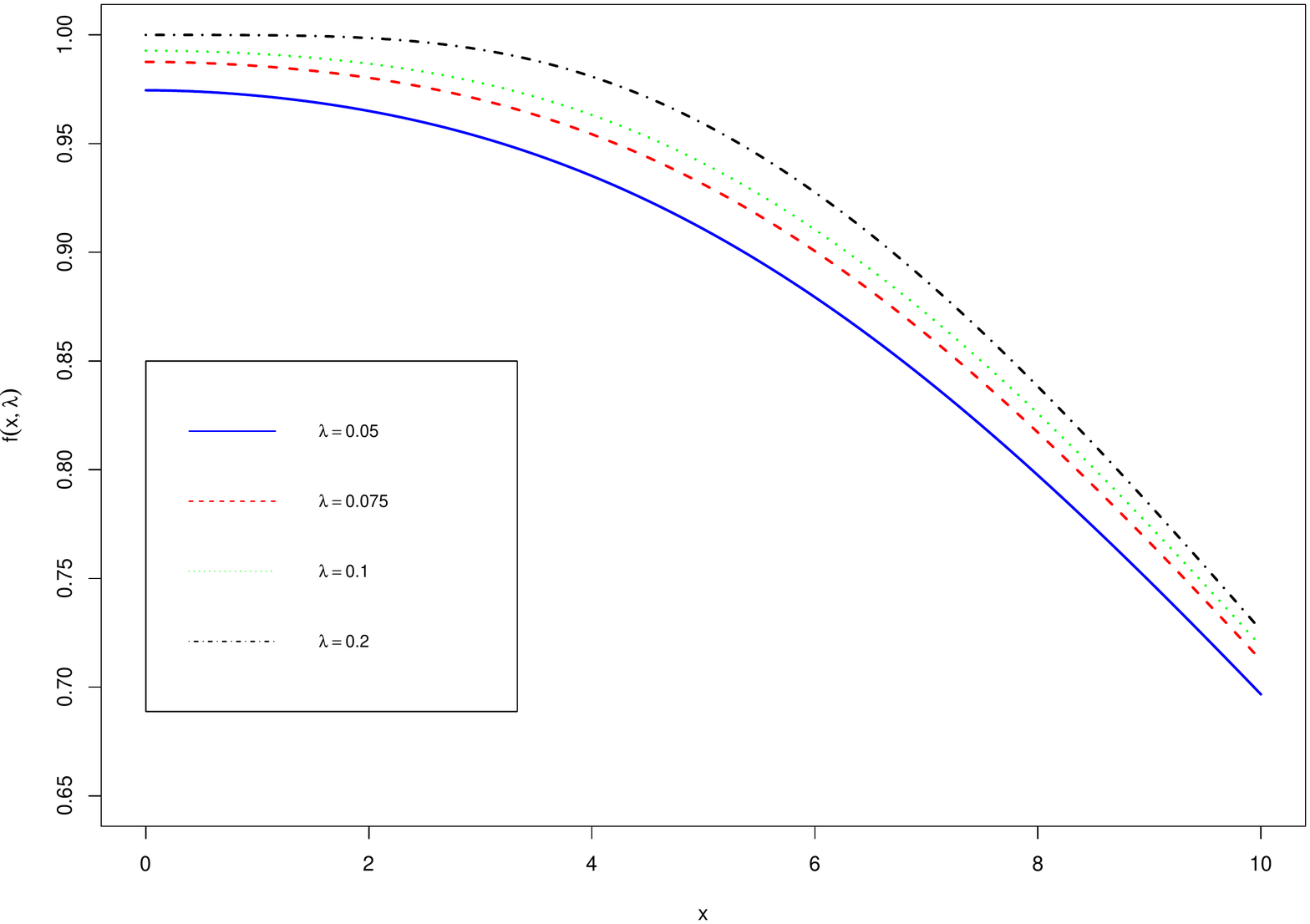}}
       {(a) $\xi(x)=0.8x$ and $\sigma=0.2$.}
&
\subf{\includegraphics[width=0.465\textwidth]{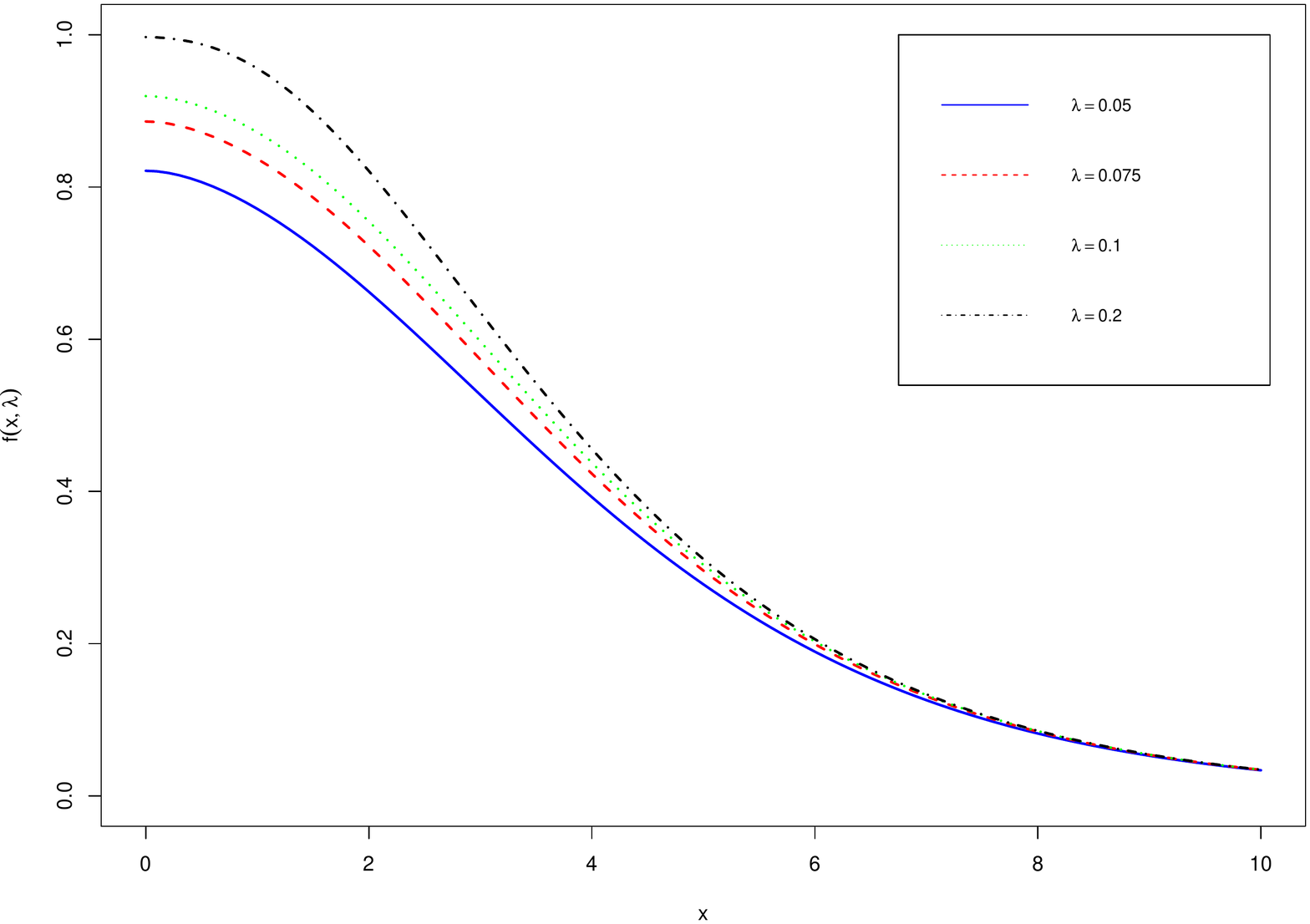}}
       {(b) $\xi(x)=0.8x$ and $\sigma=0$.}
\\
\subf{\includegraphics[width=0.465\textwidth]{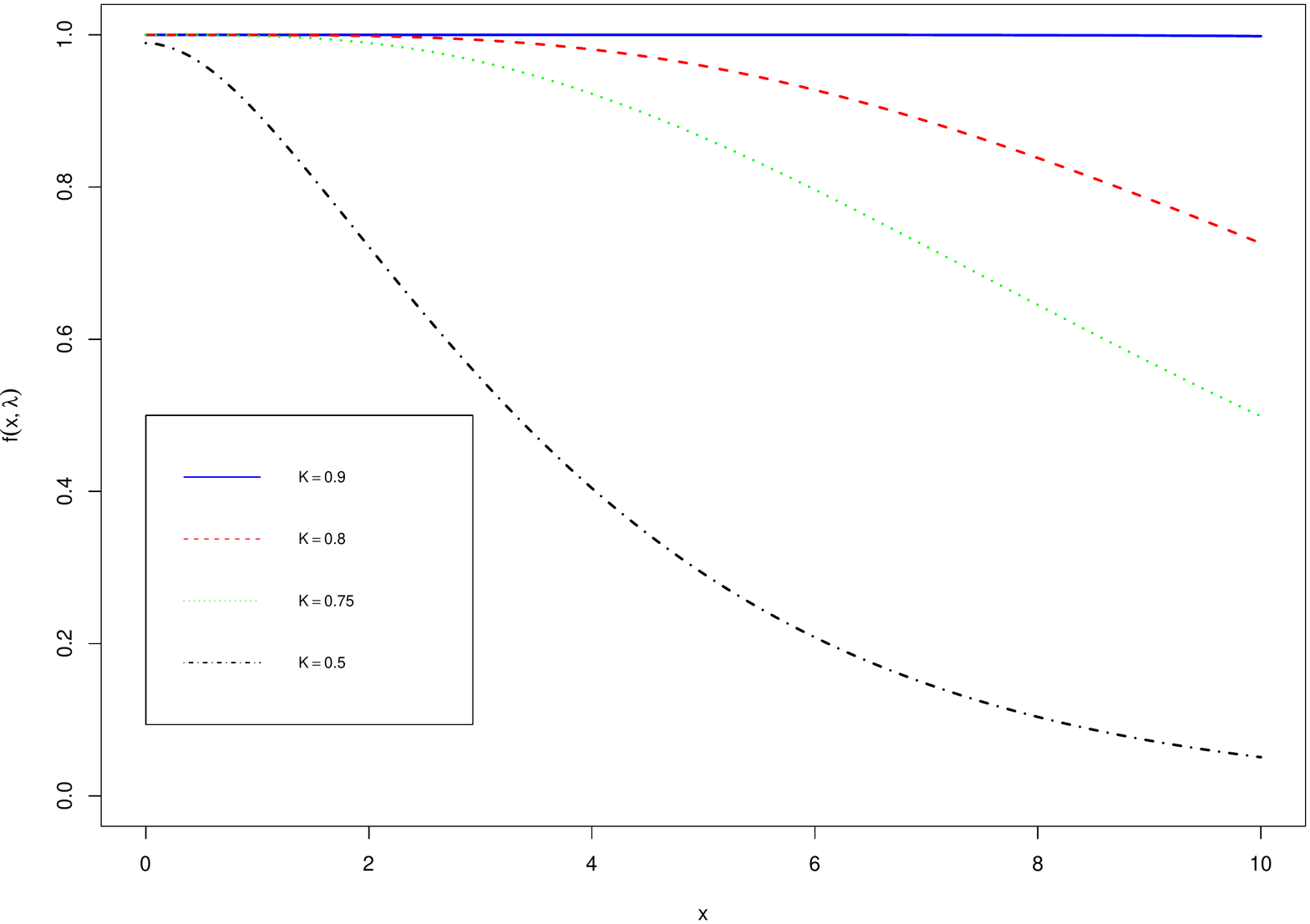}}
       {(c) $\xi(x)=Kx$, $\lambda=0.2$, and $\sigma=0.2$.}

&
\subf{\includegraphics[width=0.465\textwidth]{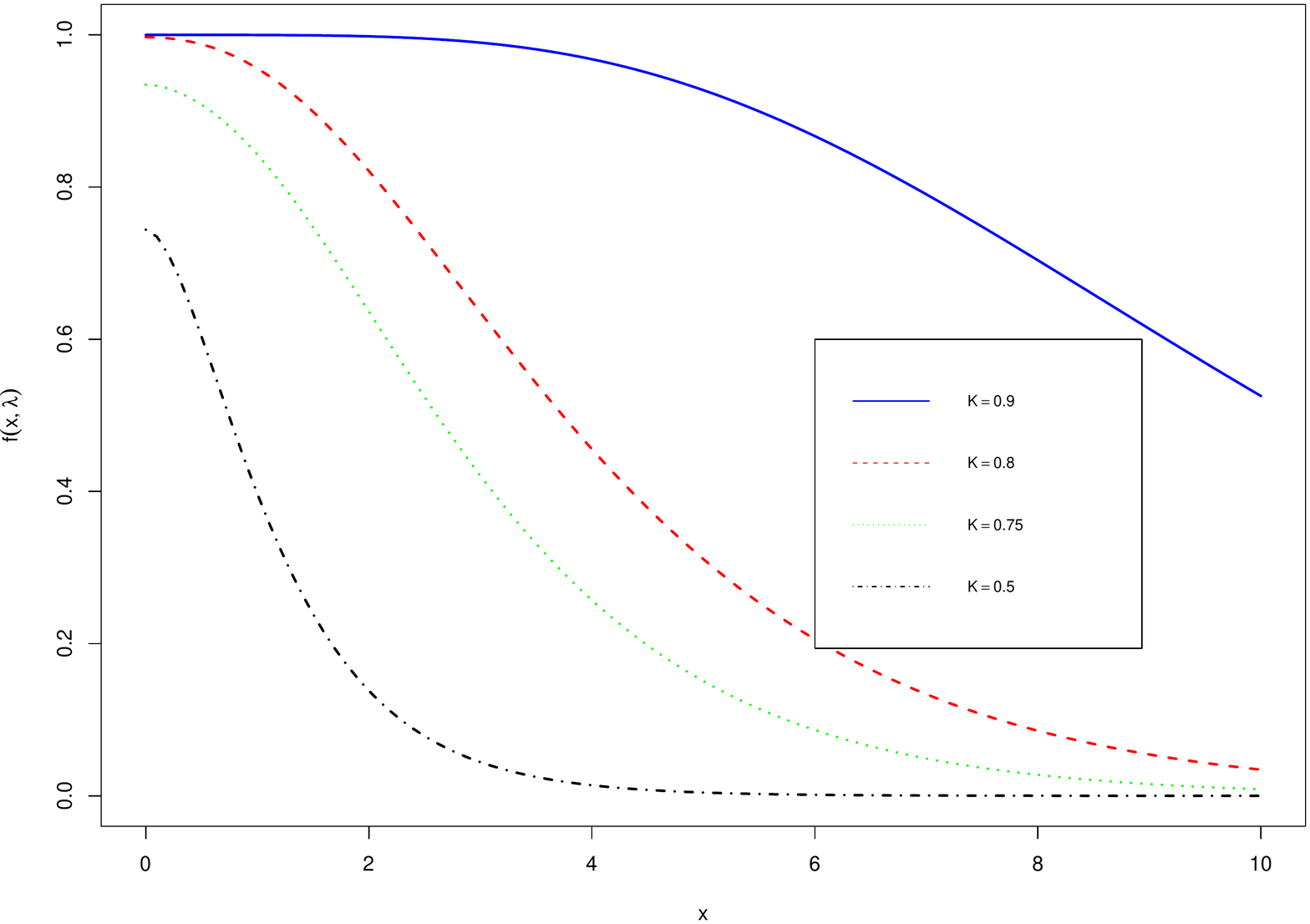}}
       {(d)  $\xi(x)=Kx$, $\lambda=0.2$, and $\sigma=0$.}
\\
\subf{\includegraphics[width=0.465\textwidth]{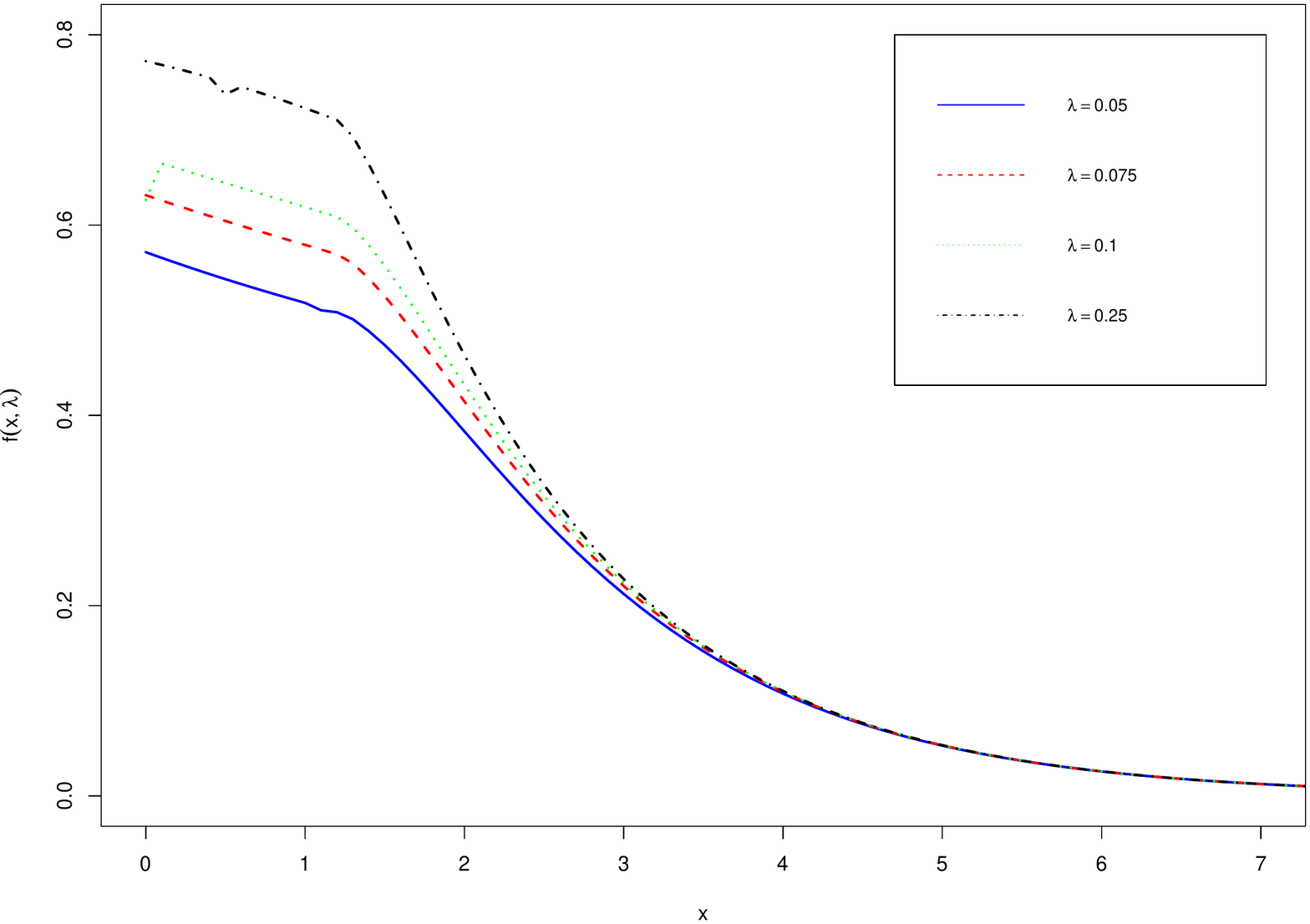}}
       {(e) $\xi(x)=\min(1,0.8x)$ and $\sigma=0.2$.}

&
\subf{\includegraphics[width=0.465\textwidth]{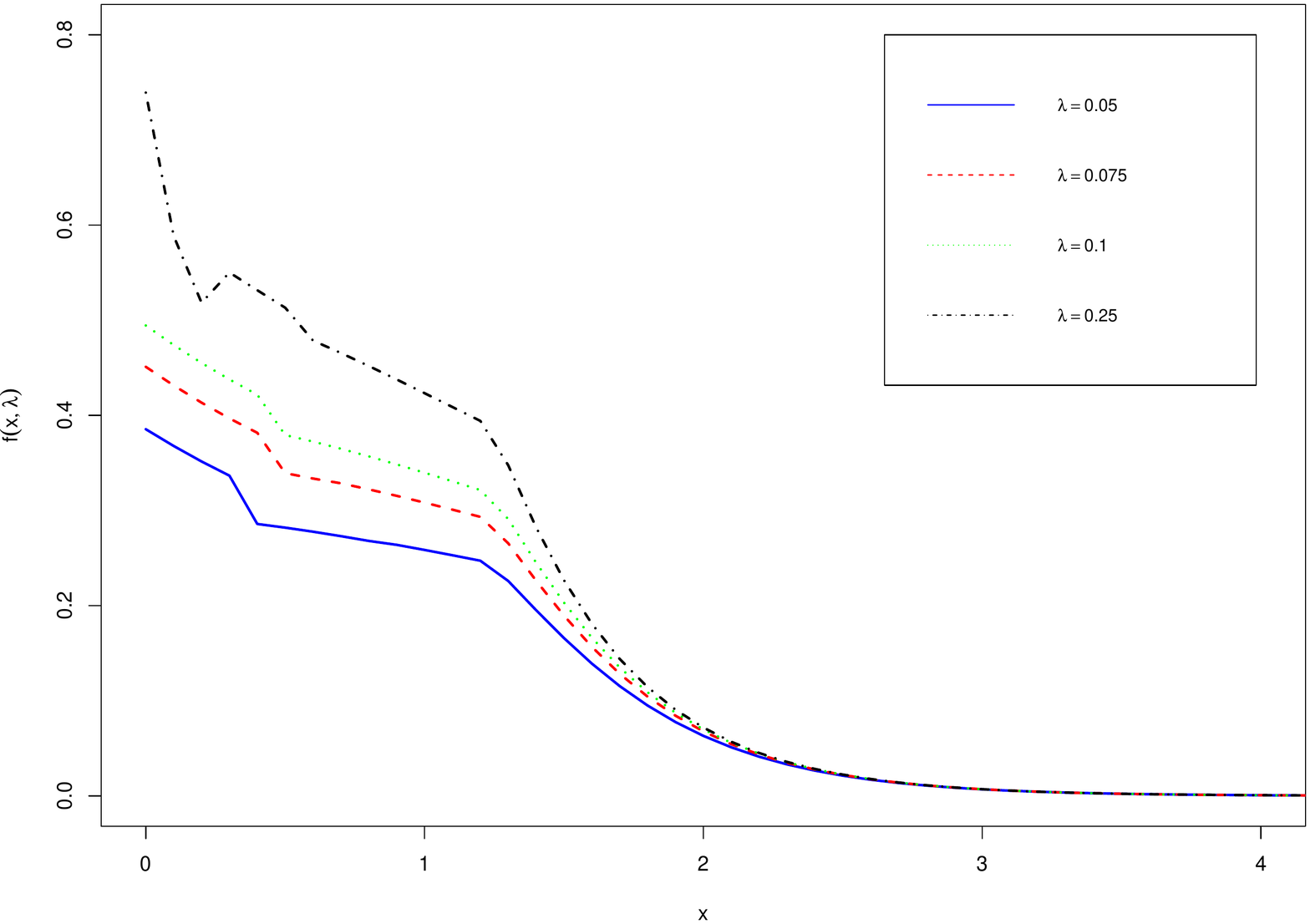}}
       {(f)  $\xi(x)=\min(1,0.8x)$ and $\sigma=0$.}
\\
\end{tabular}
\caption{Ruin probability $\mathbb{P}_x(\theta_{\lambda}<\infty)$ for downward jump-diffusion process with Laplace exponent $\psi(\theta)=\mu\theta+\frac{\sigma^2}{2}\theta^2-\frac{a\theta}{\theta+c}$ for $\mu=0.075, a=0.5, c=9, q=0.05$.}\label{fig:exit}
\end{figure}

\begin{figure}[t!]
\centering
\begin{tabular}{cc}
\subf{\includegraphics[width=0.465\textwidth]{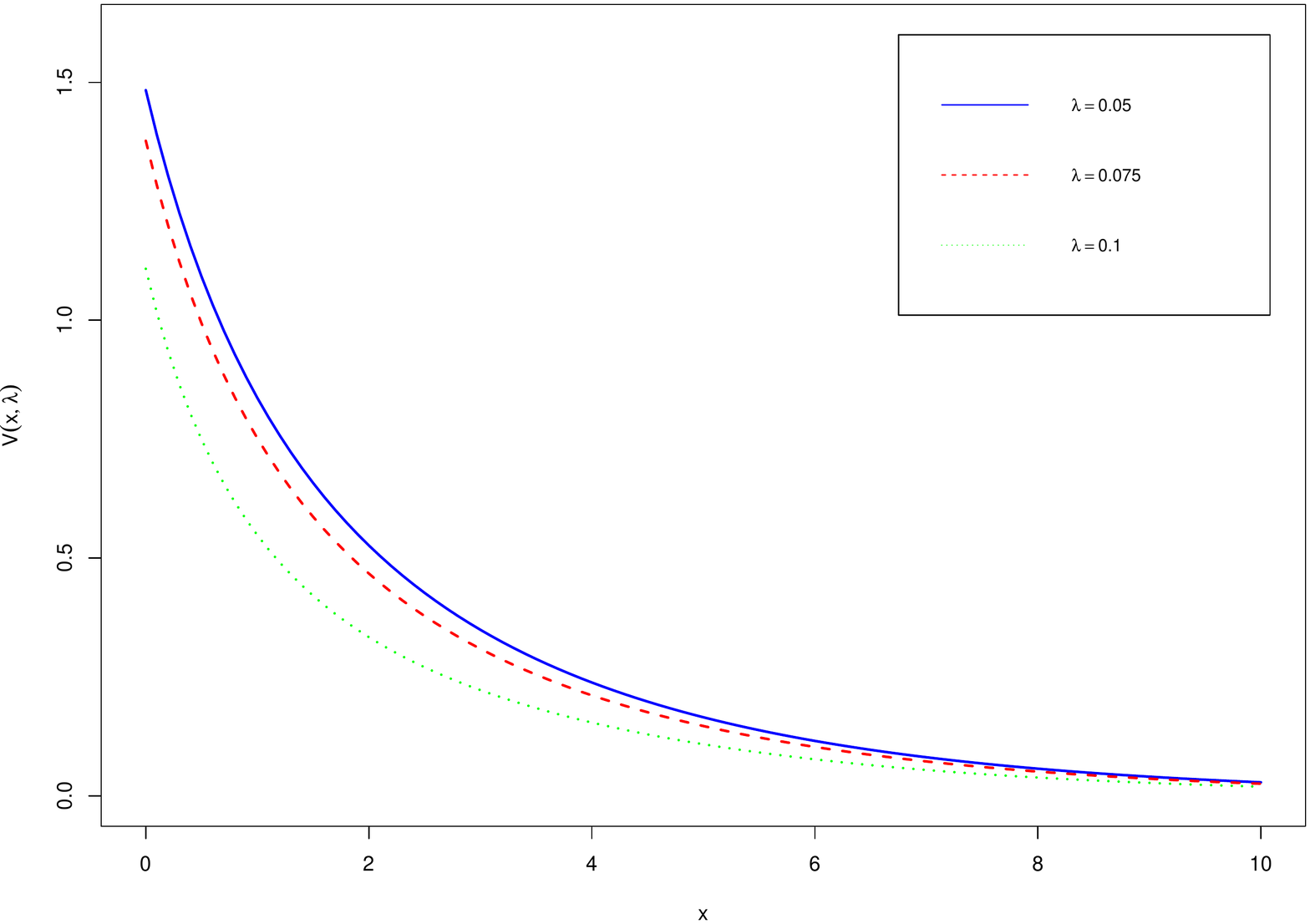}}
       {(a) $\xi(x)=0.8x$ and $\sigma=0.2$.}
&
\subf{\includegraphics[width=0.465\textwidth]{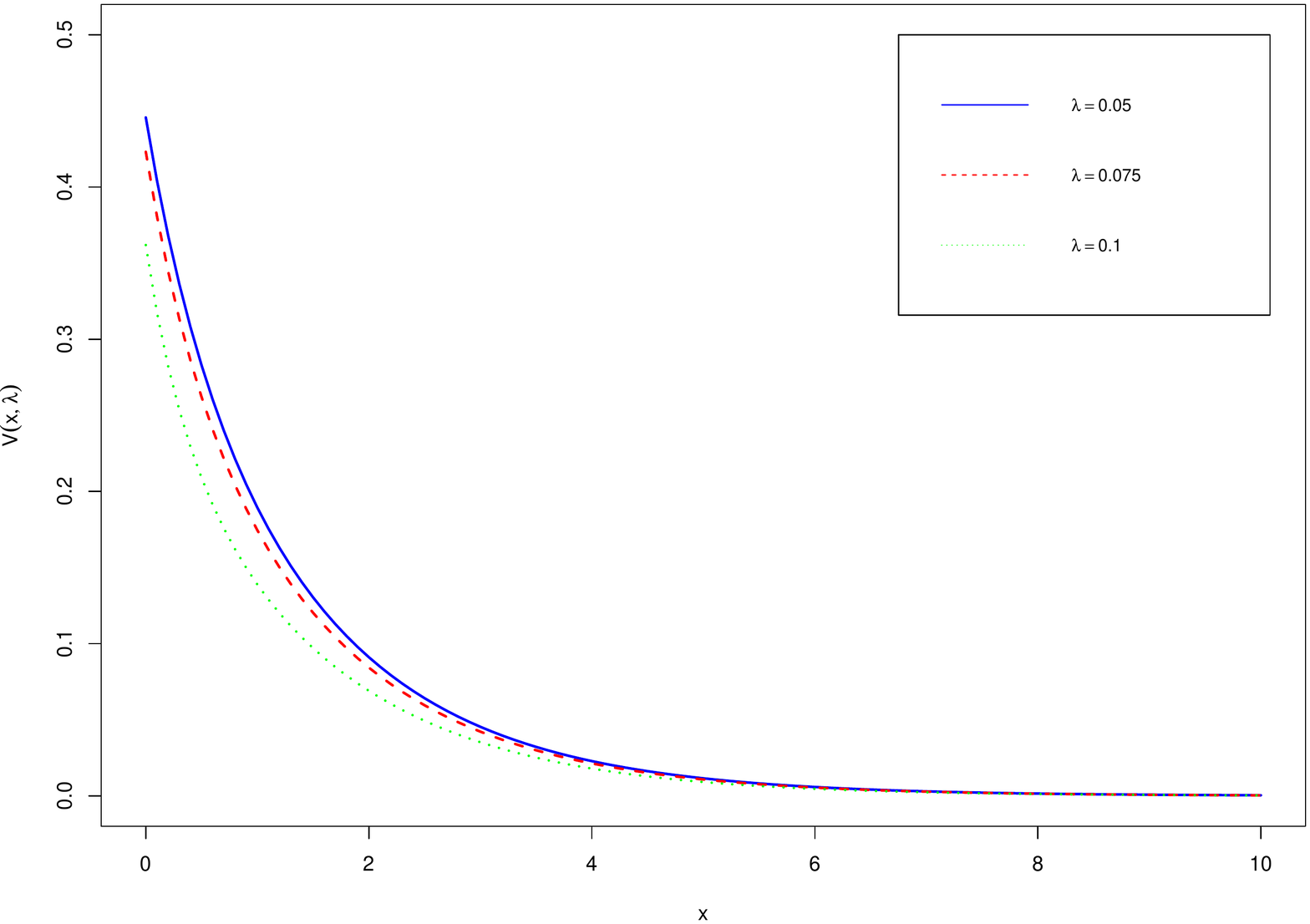}}
       {(b) $\xi(x)=0.8x$ and $\sigma=0$.}
\\
\subf{\includegraphics[width=0.465\textwidth]{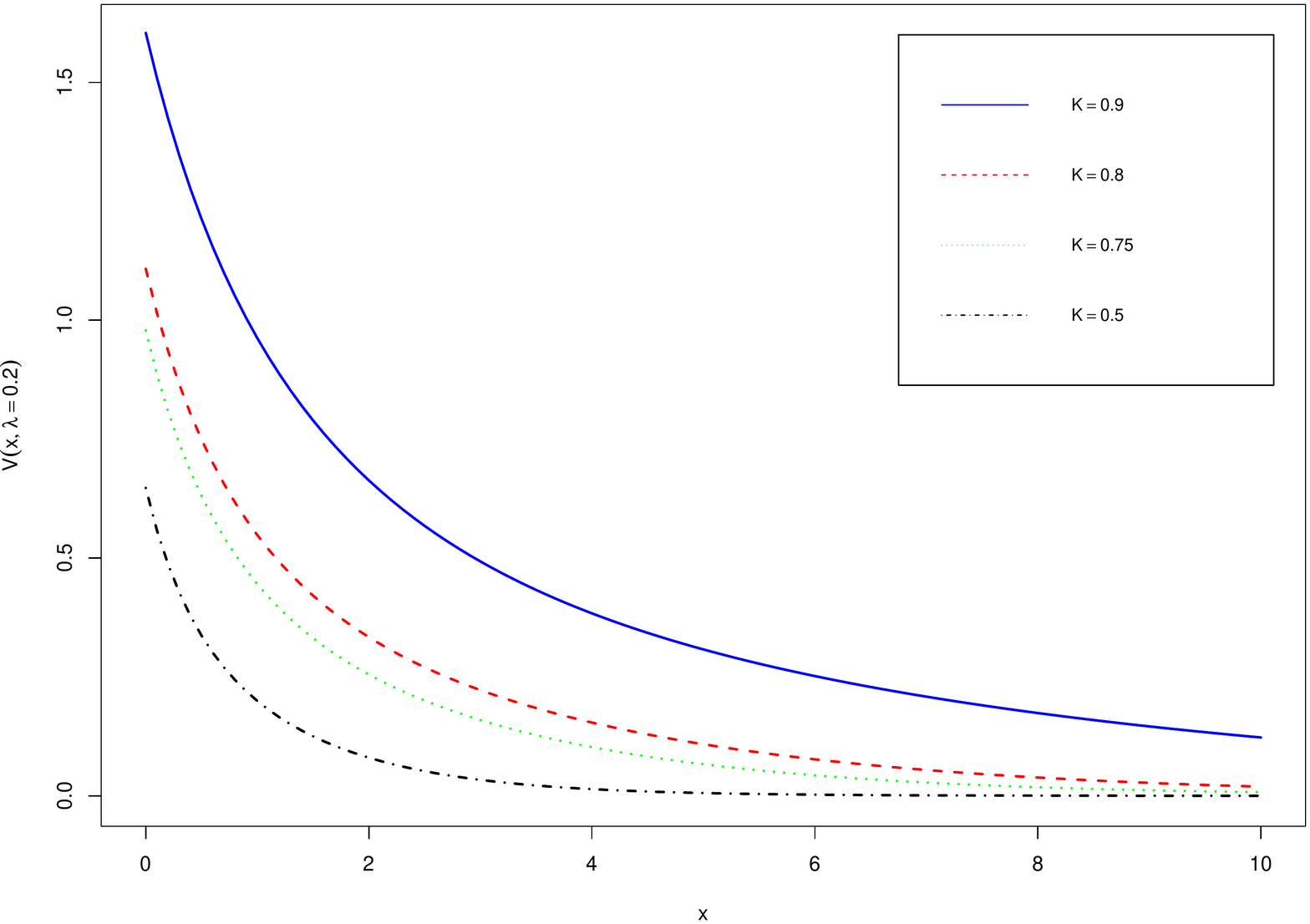}}
       {(c) $\xi(x)=Kx$, $\lambda=0.2$, and $\sigma=0.2$.}

&
\subf{\includegraphics[width=0.465\textwidth]{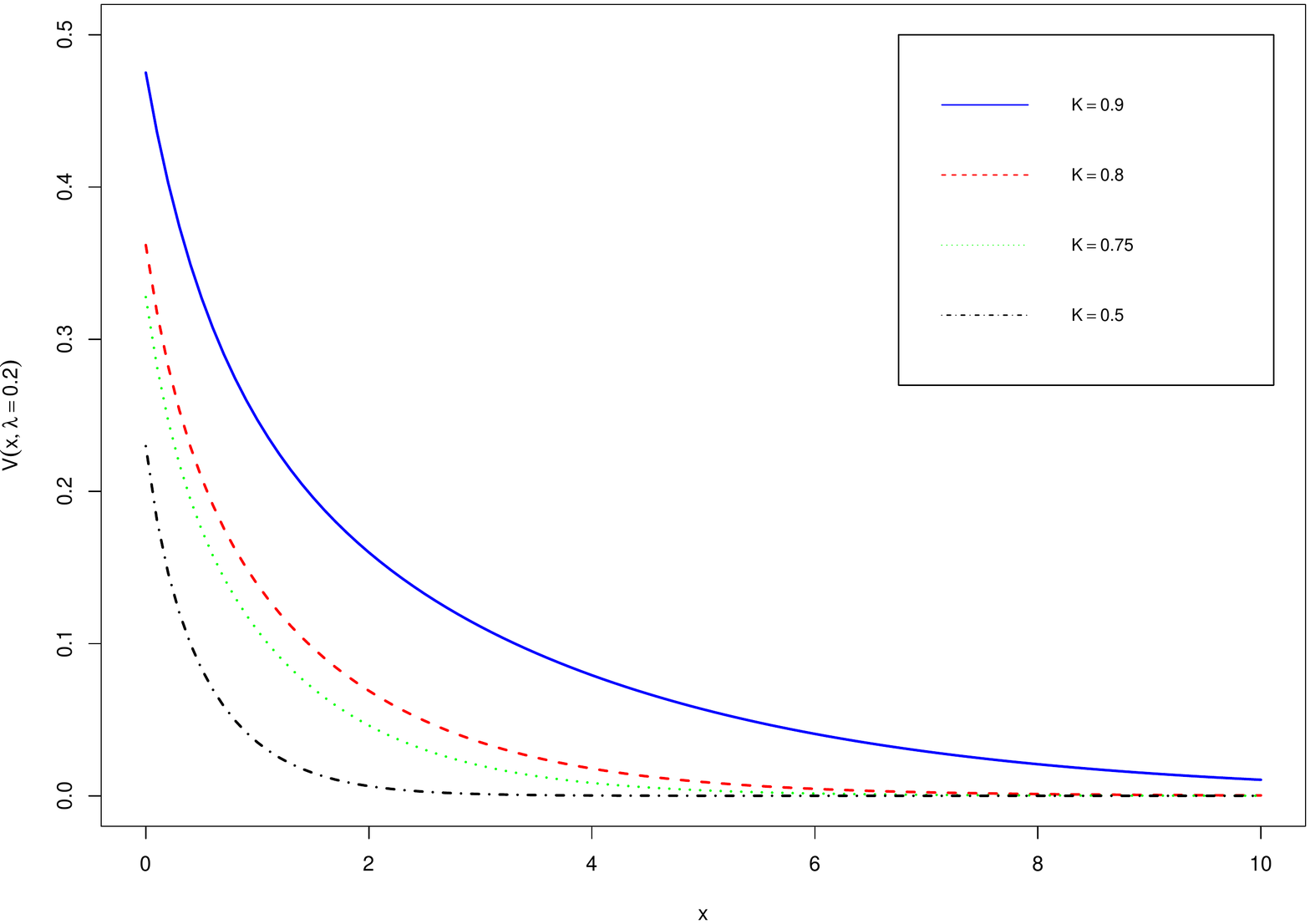}}
       {(d) $\xi(x)=Kx$, $\lambda=0.2$, and $\sigma=0$.}
\\
\subf{\includegraphics[width=0.465\textwidth]{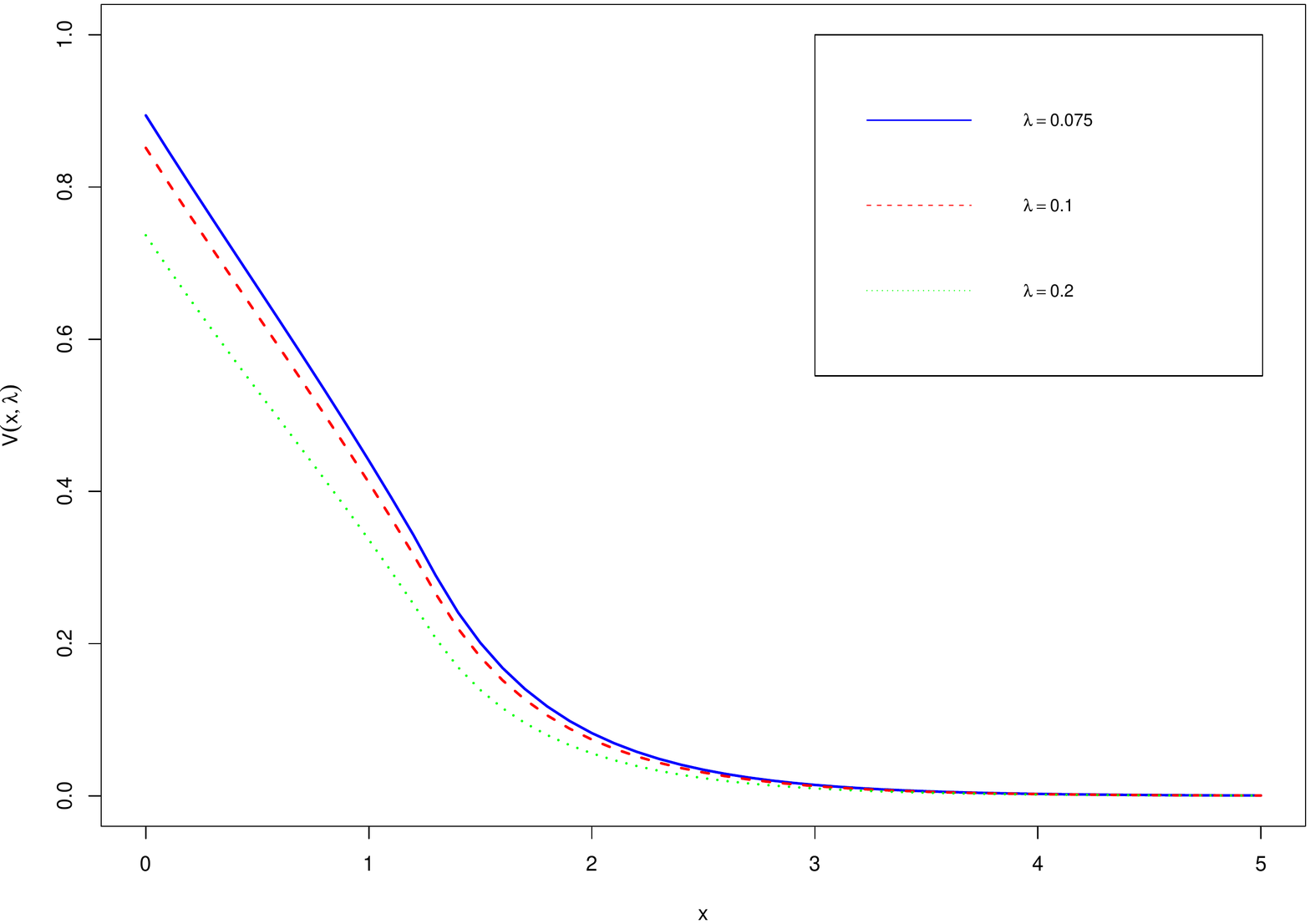}}
       {(e) $\xi(x)=min(1,0.8x)$ and $\sigma=0.2$.}

&
\subf{\includegraphics[width=0.465\textwidth]{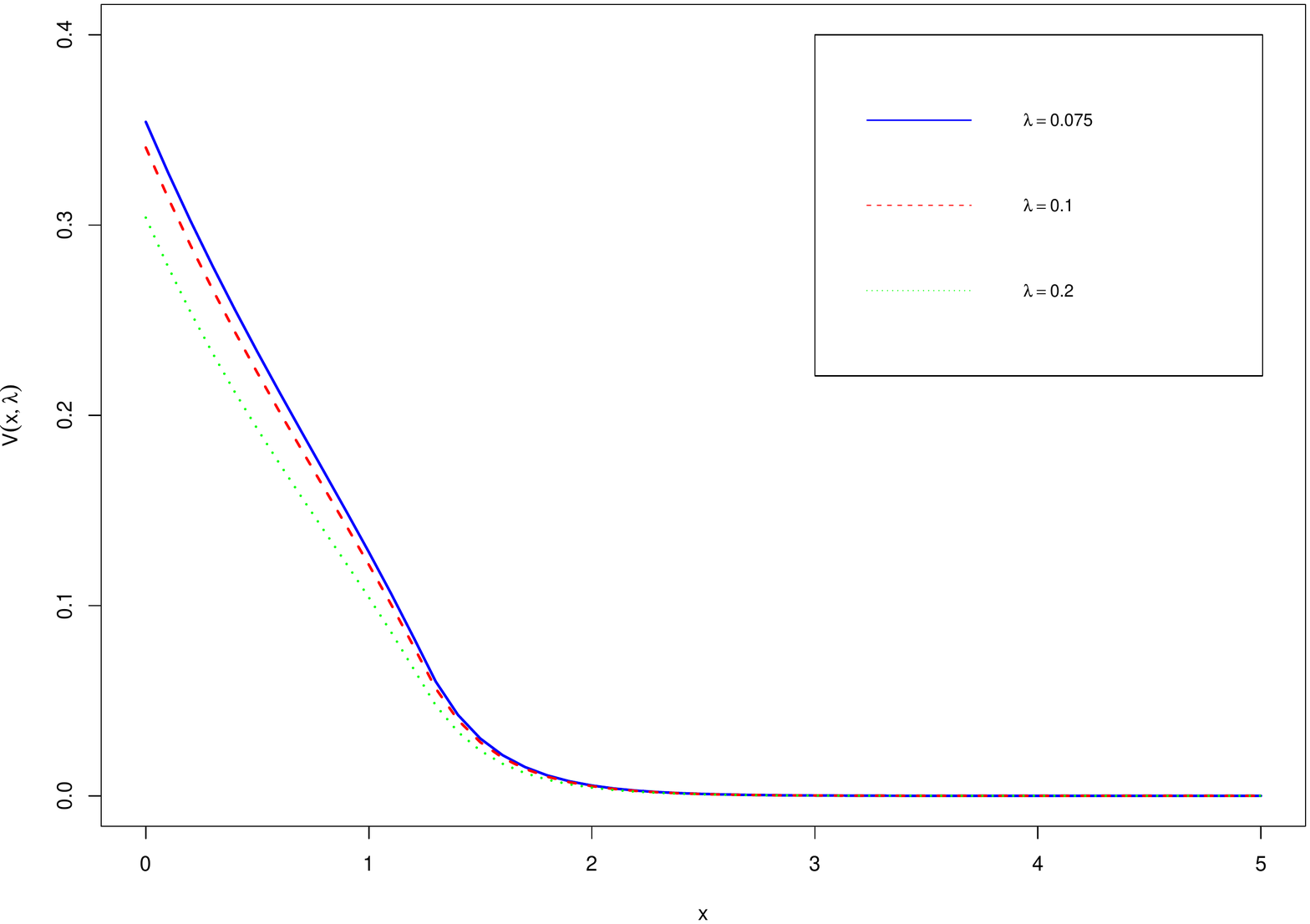}}
       {(f) $\xi(x)=min(1,0.8x)$ and $\sigma=0$.}
\\
\end{tabular}
\caption{Expected capital injection $V(x):=V_{\xi}(x;\infty)$ for jump-diffusion process with Laplace exponent $\psi(\theta)=\mu\theta+\frac{\sigma^2}{2}\theta^2-\frac{a\theta}{\theta+c}$ for $\mu=0.075, a=0.5, c=9, q=0.05$.}\label{fig:Vxi}
\end{figure}

\section{Conclusion}\label{5}
This paper presents some distributional identities concerning excursion below the last record high of surplus process, driven by downward jumps L\'evy process, with capital injection. Capital injection is provided to the firm as soon as the process goes below a drawdown level and is continuously paid until the process goes above the record or ruin occurs, which is announced at the first time the process has undertaken an excursion below the record high longer than an independent exponential period of time. Identities are given explicitly in terms of the scale function of the L\'evy process. The latter makes possible to have a fast numerical computation of the identities and do analysis on the impact of observation frequency and initial surplus to the ruin probability and the expected nett present value of the required total capital injection. The choice of some drawdown functions was made to study various shapes of the ruin probability and the nett present value of the capital injection. Numerical study shows that the results implied by the model is found to be consistent with an observation one would have in financial market. We leave this for further research.

\vspace{1cm}
\section*{Acknowledgements}
Wenyuan Wang and Xianghua Zhao
thank Concordia University where the first draft of  this paper was finished during their visits.
Wenyuan Wang acknowledges the support of the National Natural
Science Foundation of China (No. 11601197). Wenyuan Wang and Xiaowen Zhou are
supported by a National Sciences and Engineering Research Council of Canada grant
(No. RGPIN-2016-06704). Xiaowen Zhou is supported by National Natural Science
Foundation of China (No. 11771018). Part of this work was carried out while Budhi Surya was visiting Department of Satistics of  New York University Stern School of Business in September 2019. He
acknowledges Faculty Strategic Research Grant No. 20859 of Victoria University and hospitality provided by the NYU Stern.

\end{document}